\documentclass{IJCAS}
\usepackage{url}
\usepackage{array,tabularx}
\usepackage{multicol} 
\usepackage{multirow}
\usepackage{makecell}
\usepackage{booktabs}
\journalvolumn{VV}
\journalnumber{X}
\journalyear{YYYY}
\setarticlestartpagenumber{1}
\allowdisplaybreaks
\begin{document}
\title{Design of Economic Dispatch Schemes of An Isolated BESS Network Based on Distributed Discrete-time PI+Rest Consensus}
\author{Yalin Zhang*\orcid{0000-0002-3788-700X}, Zhongxin Liu\orcid{0000-0002-3565-4800}, and Zengqiang Chen\orcid{0000-0002-1415-4073}}
\begin{abstract}
	Battery energy storage systems (BESSs) are widely integrated into smart grids. For an isolated BESS network, however, capacity degradation and power loss of battery units increase operating costs. To alleviate this problem, two distributed economic dispatch (ED) schemes with discrete-time dynamics are developed in this paper, thus obtaining the optimal output power vector and ensuring supply-demand balance while considering dynamic line loss and capacity constraints. It is worth mentioning that proportional-integral protocols with reset mechanisms (PI+R protocols) are introduced into these schemes which can greatly improve the consensus rate and control accuracy. Specifically, a marginal cost (MC) consensus controller and an average power mismatch estimator are the core components of each scheme, where these two are coupled to each other. In this regard, the difference between the two schemes is that for the MC consensus controllers, one involves incorporating the estimated average power mismatch feedback term into the closed-loop system error to design the PI+R protocol, while the other does not. In addition, regarding to each PI+R protocol, the integral term is reset to 0 when the proportional term experiences zero crossing, in order to accelerate the convergence rate and reduce overshoot. The effectiveness, conditions under which the reset mechanism works, and stability of these schemes are all well analyzed. Finally, some simulation cases are designed and compared with an existing solution. From the simulation results, it can be seen that the designed second scheme greatly improves the performance of the previous scheme in consensus rate, convergence rate, BESS/agent plug and play, load switching, and wide area system application.
\end{abstract}

\begin{keywords}
	Battery energy storage system, multi-agent systems, distributed consensus control, economic dispatch, a PI+Reset controller.
\end{keywords}

\maketitle

\makeAuthorInformation{Yalin Zhang is with the College of Electrical Engineering, Zhejiang University, Hangzhou 310012. Zhongxin Liu and Zengqiang Chen are with the College of Artificial Intelligence, Nankai University, Tianjin 300350, and also with the Tianjin Key Laboratory of Interventional Brain-Computer Interface and Intelligent Rehabilitation, Nankai University, Tianjin 300350, China. (e-mail: zhangyalin@zju.edu.cn; lzhx@nankai.edu.cn; chenzq@nankai.edu.cn).

* Corresponding author.
}

\runningtitle{2025}{Yalin Zhang, Zhongxin Liu, and Zengqiang Chen}{Manuscript Template for the International Journal of Control, Automation, and Systems: ICROS {\&} KIEE}{xxx}{xxxx}{x}
\begin{center}
	N\footnotesize{OMENCLATURE}
\end{center}
\normalsize
\begin{tabbing}
	\hspace{2cm} \= \kill
	BESS\> battery energy storage system\\
	ED\>economic dispatch\\
	MG\> micrigrid\\
	MAS\> multi-agent system\\
	MC\> marginal cost\\
	$\alpha_i$, $\beta_i$\>the cost function coefficients of\\
	$P_i^k$, $f_i$\>the output power and cost function of\\
	$\lambda_i^k$\>the MC of BESS $i$\\
	$P_{\mathrm{D},i}$, $P_{\mathrm{L},i}$\> the load power and the line power loss\\
	$P_{\mathrm{m},i}$, $P_{\mathrm{M},i}$\> the output power upper and lower limits\\
	$n$\>the number of BESSs\\
	$Z_i$\>the line power loss ratio\\
	${\cal G}$, $\cal A$\> the communication topology\\
	\>and its adjacency matrix\\
	$L$\>the Laplacian matrix\\
	$\eta_i(L)$, $\eta_M$\>the eigenvalues and their maximum value\\
	\>of $L$\\
	$\Delta {\hat P}_{a,i}^k$\> the estimated average power mismatch\\
	$\sigma^k$\>a decay feedback gain\\
	$\xi_{\lambda,i}^k$, $\nu_{\lambda,i}^k$\>the control error and its accumulate of $\lambda_i^k$\\
	$\xi_{P,i}^k$, $\nu_{P,i}^k$\>the control error and its accumulate of $\Delta {\hat P}_{a,i}^k$\\
	$\Delta P_i^k$\> the power mismatch\\
	$h_1$ and $h_2$\>the control gains of $\lambda_i^k$\\
	$z_1$ and $z_2$\>the control gains of $\Delta {\hat P}_{a,i}^k$\\
	$\lambda^k$, $\xi_{\lambda}^{k}$, and $\nu_{\lambda}^{k}$\>the stack vectors for $\lambda_i^{k}$, $\xi_{\lambda,i}^{k}$, and $\nu_{\lambda,i}^{k}$\\
	$\Delta\hat{P}_{a}^{k}$ and $\xi_{P}^{k}$\>the stack vectors for $\Delta\hat{P}_{a,i}^{k}$ and $\xi_{P,i}^{k}$\\
	$\nu_{P}^{k}$ and $\Delta P^{k}$\>the stack vectors for $\nu_{P,i}^{k}$ and $\Delta P_i^{k}$\\
	${\cal J}_\lambda$, ${\cal J}_P$\>the jump sets\\
	${\cal F}_\lambda$, ${\cal F}_P$\>the flow sets\\
	$\Phi(L)$\> the matrix of the base system
\end{tabbing}
\section{Introduction}
Recently, a high proportion of renewable energy is used for power generation in smart grids, and its unique randomness and intermittency, however, greatly reduce the quality of electricity \cite{11075888, caleroReviewModelingApplications2022, 10660524}. For this reason, battery energy storage systems (BESSs) are extensively integrated into smart grids, as shown in Fig. \ref{figei}, to suppress supply-demand imbalances and serve as backup power sources, thereby providing more reliable electricity \cite{wangApplicationEnergyStorage2022, caleroReviewModelingApplications2022, yangModellingOptimalEnergy2022, hossainlipuReviewControllersOptimizations2022}. At the same time, some worrying facts have to be noted.
During the charging and discharging circle, the capacity of each battery continuously deteriorates \cite{caleroReviewModelingApplications2022, forero-quinteroProfitabilityAnalysisDemandside2022, rouholaminiReviewModelingManagement2022}, which makes operators have to replace battery units frequently. In addition, the power consumed by the internal impedance of each battery also increases operating costs \cite{caleroReviewModelingApplications2022, forero-quinteroProfitabilityAnalysisDemandside2022, rouholaminiReviewModelingManagement2022}. Therefore, it is necessary to design an efficient economic dispatch (ED) scheme for a BESS network to reduce operational expenses.
\begin{figure}
	\centering
	\renewcommand{\thefigure}{1}
	\includegraphics[width=8cm]{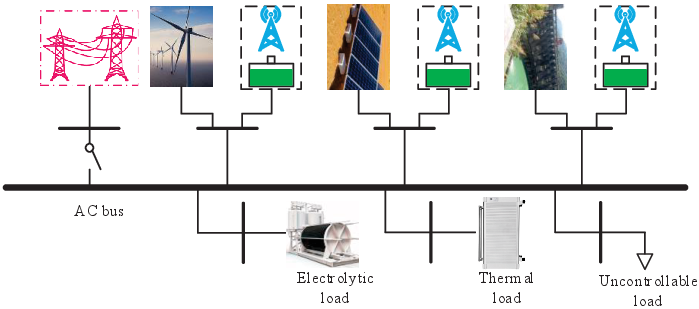} \caption{BESSs in microgrids\label{figei}}
\end{figure}
\par At present, the ED solutions for microgrids MGs mainly include centralized and distributed approaches. The implementation of a centralized solution is simple and feasible, but the disadvantage of being susceptible to single point failure is also evident \cite{caleroReviewModelingApplications2022, rouholaminiReviewModelingManagement2022,bazhang2022}. In addition, the expensive centralized controller and their high bandwidth communication links greatly increase operational expenses \cite{caleroReviewModelingApplications2022, rouholaminiReviewModelingManagement2022,bazhang2022}. Given this, we plan to invest effort in designing a distributed ED scheme for an isolated BESS network in this article.
\par As a typical distributed technology, multi-agent systems (MASs) have strong robustness and scalability. Most of the existing distributed ED solutions are based on MASs, and further, they need to combine methods such as Lagrange multiplier method \cite{zhangConvergenceAnalysisIncremental2012}, gradient search method \cite{Yi2016}, and Newton method \cite{andersonDistributedApproximateNewton2019}, etc. Correspondingly, some distributed schemes such as consensus-based \cite{zhangConvergenceAnalysisIncremental2012}, gradient descent \cite{Yi2016}, and approximate Newton methods \cite{andersonDistributedApproximateNewton2019} have been developed. However, current distributed schemes based on gradient descent and approximate Newton methods do not take into account both dynamic line loss and output limitations of distributed generators. On the contrary, consensus-based solutions, due to their simple structure and ease of modeling, have produced a large number of excellent results under these physical limitations. The consensus-based distributed ED schemes are mainly divided into control schemes with continuous time dynamics and discrete time dynamics.
\par \textit{1). On continuous-time schemes}. For example, in order to save communication resources in a distributed ED scheme, some distributed solutions with event-triggered-based communication mechanisms are developed in \cite{Liu2022b, 9928208}. In addition, in order to accelerate system convergence, some distributed finite \cite{Liu2022b, 10210483}/fixed \cite{Dai2021, Li2022m, Huang2023b}/prescribed \cite{Ji2023} time ED schemes are developed. These solutions with continuous time dynamics effectively solve the ED problem of microgrids by designing marginal cost (MC) consensus controllers and average power mismatch estimators. 
\par \textit{2). On discrete-time schemes}. At present, the consensus-based distributed discrete-time ED schemes mainly focus on solving problems such as privacy protection, communication network attacks, and packet loss, etc. Some researchers in \cite{yeRandomWeightPrivacyPreservingAlgorithm2021} use random communication weights to design a communication mechanism to improve the privacy of communication data, and design an error compensation mechanism to improve the scheme. A differential privacy protection mechanism based on double-layer communication has been developed in \cite{wangDisEHPPCEnablingHeterogeneous2022a}. In addition, the authors developed an event-triggered privacy protection method in \cite{Yan2022a}, which improves the privacy of communication data while saving communication resources. In order to cope with network attacks, virtual power plant based and local event driven solutions have been developed in \cite{liRobustDistributedEconomic2018} and \cite{sahooLocalizedEventDrivenResilient2021a}, respectively. In addition, a push-based method is developed in \cite{wangPushBasedDistributedEconomic2021} to resist gradient noise and communication delay, and is suitable for unbalanced communication networks. In order to alleviate the problem of data packet loss in communication networks, algorithms such as gossip \cite{zhangRobustDistributedSystem2016} and distributed algorithms based on buffer nodes \cite{liConsensusBasedEnergyManagement2023} have been proposed to ensure the normal transmission of power mismatch information.
\par Inspired by the solution to solve the distributed ED problem in microgrids, some researchers design some distributed solutions based on MASs for BESS networks \cite{yuFrequencySynchronizationPower2021b, zhaoDifferentialPrivacyEnergy2022, jinManageDistributedEnergy2022}. A distributed ED scheme considering battery capacity limitations is designed for an isolated BESS network in \cite{yuFrequencySynchronizationPower2021b}. Unfortunately, the line loss in a BESS network is overlooked. Even worse, the cost function of a BESS has not been effectively modeled, but instead used that of other types of microsources. A similar
problem also appears in \cite{zhaoDifferentialPrivacyEnergy2022}, where the authors design a distributed ED scheme with a privacy protection mechanism. It is not until in \cite{jinManageDistributedEnergy2022} that the authors initially solved this problem by considering the cost of charging and discharging losses in
the cost function. 
\par Based on the above analysis, we notice that currently, distributed ED schemes with discrete time dynamics for BESSs/MGs are all proportional control protocols \cite{yuFrequencySynchronizationPower2021b, zhaoDifferentialPrivacyEnergy2022, jinManageDistributedEnergy2022}. These schemes are validated to be effective. However, these protocols exhibit low control accuracy and slow convergence rates, which gives us room to expand our capabilities. As is well known, a PI controller exhibits high control accuracy but may cause overshoot \cite{banosResetControlSystems2011}.
\par In view of this, some researchers have started to improve the PI controller by introducing a reset mechanism \cite{chengResetControlLeaderfollowing2022, huResetControlConsensus2022, mengResetControlSynchronization2019}. We have also designed a continuous time MC consensus controller with a reset mechanism for a distributed ED problem in a grid-connected/isolated BESS network in \cite{10302354, 10488096}. Similar to a finite/fixed/preset time control scheme, we design a distributed ED scheme with continuous time dynamics and a simple and easy to implement PI+Reset controller in these two works to accelerate consensus rate, convergence rate, and control accuracy. Unfortunately, due to previous technological limitations, capacity limitations and dynamic line losses are not taken into account in our previous work. These two physical limitations resulted in a nonlinear and strongly coupled relationship between variables. Considering both of these limitations, the current distributed solutions with continuous time dynamics cannot effectively solve the ED scheme of isolated microgrids. Some researchers conduct a preliminary investigation in \cite{chenDistributedEconomicDispatch2021} into a distributed scheme with discrete-time dynamics that consider both constraints simultaneously. However, the coupling term is not introduced into the control error, and with the application of a P-controller, the solution of ED encounter poor dynamic performance. In view of this, we plan to design a distributed ED scheme with reset-mechanism-based discrete-time dynamics to solve the ED problem of an isolated BESS network with capacity constraints and dynamic line losses. In detail, our contributions are listed as follows:
\begin{enumerate}
	\item In order to achieve optimal power output of each BESS, a distributed ED scheme with reset mechanisms is developed, which ensure MC consensus and the estimated average power mismatch convergences to 0, thereby ensuring supply-demand balance. Compared to the existing solution in \cite{chenDistributedEconomicDispatch2021}, the designed approach exhibits significant advantages in terms of consensus,
	stability, and dynamic performance.
	\item Furthermore, an improved solution is proposed. The feedback term of the estimated average power mismatch in the MC consensus control scheme is incorporated into the closed-loop error to further improve dynamic performance, which is different from the one in \cite{chenDistributedEconomicDispatch2021}. Based on this, the PI+R controller is redesigned. Through simulation results, it can be seen that the dynamic performance of MC and the estimated average power mismatch is greatly improved compared to the previous scheme.
	\item The gain conditions for ensuring the stability of the base system, i.e. regularity, are given, thereby deriving the stability condition for the distributed PI+R controller. The enabling conditions for the reset mechanism are
	also derived. In addition, the stability of MC and the estimated average power mismatch under the designed scheme is cleverly demonstrated by utilizing the characteristics of the cost function and power constraints, which is rarely reported in the literature on distributed ED schemes considering transmission losses.
\end{enumerate}
\par The remaining part of this article is arranged as follows. The ED problem of an isolated BESS network is modeled in the section 2, and the optimal working point of the network is derived from it. In the section 3, two distributed ED schemes with reset mechanisms are constructed and their stability, regularity, and effectiveness are analyzed. In order to test the performance of the designed schemes, some simulation cases are designed in the section 4. Finally, a conclusion is summarized in the section 5.
\section{Preliminaries}
In this section, the optimal output power vector is obtained by applying the Lagrange multiplier method. Subsequently, the control objectives of this article are provided, and the upcoming graph theory is introduced.
\subsection{The Optimal ED for A BESS Network}
\par Taking into account the internal consumption cost and
capacity degradation cost of a battery, the following convex
cost function \cite{10302354, 10488096, Zhang02092025} for BESS $i$ is constructed as
$$f_i=\beta_iP_i^2+\alpha_iP_i,\;i\in\{1,2,\cdots,n\},$$
where $\beta_i=(1+m_i)\pi_i$, $\alpha=(1+m_i)\Omega_i$, $m_i$, $\pi_i$ and $\Omega_i$ are constants and defined in \cite{10302354, 10488096, Zhang02092025}. So, the ED objective
function and its constraints for an isolated BESS network are
$$\begin{array}{l}
	\mathrm{min}\;F=\sum_{i=1}^{n}f_i,\\
	s.t.\;\sum_{i=1}^{n}P_i=\sum_{i=1}^{n}(P_{D,i}+P_{L,i}),
	\;P_{m,i}\le P_i\le P_{M,i},
\end{array}$$
where $P_{D,i}$ is the total active load at bus $i$, $P_{m,i}$ and $P_{M,i}$ are the output power upper and lower limits of BESS $i$. For a bus without any BESS, it is considered to be connected to a BESS with both upper and lower output power limits of 0, such as bus $j$, $P_{M,j} = 0$ and $P_{m,j} = 0$.
\par Construct a Lagrangian function as follows
$$\begin{aligned}
	L_a(P_1,\cdots,P_n)=&F+\omega_1(\sum_{i=1}^{n}(P_{D,i}+P_{L,i})
	-\sum_{i=1}^{n}P_i)\\
	&+\sum_{i=1}^{n}\omega_{2,i}(P_i-P_{M,i})\\
	&+\sum_{i=1}^{n}\omega_{3,i}(P_{m,i}-P_i),
\end{aligned}$$
where $\omega_1$, $\omega_{2,i}$, and $\omega_{3,i}$ are Lagrangian multipliers. Calculate the gradient of the above as follow
$$\frac{\partial L_a(P_1,\cdots,P_n)}{\partial P_i}=2\beta_iP_i+\alpha_i+\omega_1(2Z_i P_i-1),$$
which leads to the optimal condition as follow by letting $\frac{\partial L_a(P_1,\cdots,P_n)}{\partial P_i}=0$
\begin{equation}
	\label{gsmmc}
	\lambda_1=\cdots=\lambda_n=\omega_1^*,
\end{equation}
where $\lambda_i=\frac{2\beta_iP_i+\alpha_i}{1-2Z_iP_i}$ for $i\in\{1, 2, \cdots, n\}$ is the MC of BESS $i$, $\omega_{1}^{*} = \frac{\sum_{i=1}^{n}(D_{i}+\frac{\alpha_{i}}{2\beta_{i}})}{\sum_{i=1}^{n}\frac{1}{2\beta_{i}}}$. Combined with the capacity limitations, the optimal output power of BESS $i$ can be calculated from the following,
\begin{equation}
	\label{power}
	P_i^*=\left\{\begin{matrix}
		P_{m,i},&if\;\frac{\omega_1^*-\alpha_i}{2(\beta_i+Z_i\omega_1^*)}<P_{m,i},\\
		\frac{\omega_1^*-\alpha_i}{2(\beta_i+Z_i\omega_1^*)},&if\;P_{M,i}\le\frac{\omega_1^*-\alpha_i}{2(\beta_i+Z_i\omega_1^*)}\le P_{m,i},\\
		P_{M,i},&if\;\frac{\omega_1^*-\alpha_i}{2(\beta_i+Z_i\omega_1^*)}>P_{M,i}.
	\end{matrix}
	\right.
\end{equation}
\begin{figure}
	\centering
	\includegraphics[width=7cm]{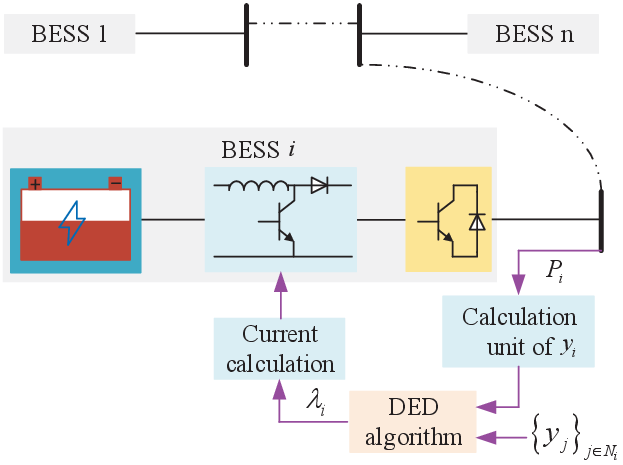} \caption{The ED framework of isolated BESSs\label{figisolate}}
\end{figure}
\subsection{Control Objectives}
\par For an isolated BESS network, its load needs to be shared by the battery units, unlike in grid connected mode where the utility grid compensates for power mismatch. Therefore, for the ED problem of an isolated BESS network, it is necessary to promote MC consensus to reduce cost, while ensuring that the output power of each BESS does not exceed the limitations and meeting supply-demand balance. Based on the previous analysis, we have summarized these goals based on a ED framework shown in Fig. \ref{figisolate} in detail as follows:
\begin{itemize}
	\item For the isolated mode, MCs of BESSs reach consensus
	asymptotically, i.e.,
	$$\lim\limits_{k\to+\infty}\vert\lambda_i-\lambda_j\vert=0,\;\forall i\in\{1,2,\cdots,n\}.$$
	\item For the isolated mode, the power calculated by \eqref{power} can meet the required total power containing line loss and load power, i.e.,
	$$\sum_{i=1}^{n}(P_{D,i}+P_{L,i}-P_i)=0.$$
\end{itemize}
\subsection{Graph Theory}
\par In this article, an isolated BESS network is governed by MASs. For a network containing $n$ BESSs, there are $n$ agents in this MAS corresponding to each BESS in this network. The communication topology between agents can be modeled as ${\cal G}({\cal V},{\cal E})$, where $v_i$ for $i\in\{1,2,\cdots,n\}$ and ${\cal V}=\{v_1,v_2,\cdots,v_n\}$ represent its vertices and their sets, $(v_i,v_j)$ and ${\cal E}$ represent its edges and their sets,
respectively. The agent $j$ for $j\in\{1,2,\cdots,n\}$ is considered
as a neighbor of agent $i$ if agent $i$ can access agent $j$, and this is denoted as $(v_i,v_j)\in{\cal E}$. The graph ${\cal G}$ is called undirected or bidirectional if $(v_i,v_j)\in{\cal E}$ inevitably leads to $(v_j,v_i)\in{\cal E}$. Two useful matrices are defined here. One is called adjacency matrix ${\cal A}=[a_{ij}]_{n\times n}$, where $a_{ij}=1$ if $(v_i,v_j)\in{\cal E}$, otherwise $a_{ij}=0$. The other is the degree matrix ${\cal D}=diag(d_1,d_2,\cdots,d_n)$, where $d_i=\sum_{j=1}^{n}a_{ij}$. So, the Laplace matrix corresponding to $\cal G$ is calculated as $L={\cal D}-{\cal A}$.
\begin{assumption}
	\label{assu1}
	Assuming that $\cal G$ used in this article is undirected and connected, then $L$ is symmetric and has a unique simple zero eigenvalue.
\end{assumption}
\begin{definition}
	$\eta_i(L)$ denotes taking the $i$-th eigenvalue of $L$ such that $0=\eta_1(L)\le\eta_2(L)\le\cdots\le\eta_n(L)=\eta_M$.
\end{definition}
\section{Distributed PI+Reset Schemes for Discrete Time Economic Dispatch of An Isolated BESS Network}
In this section, two distributed ED schemes with PI+R controllers are provided. The base system, consensus, and stability of the two schemes are analyzed, and the effectiveness of the designed schemes are explained accordingly.
\subsection{Distributed Control Scheme for Achieving Control Objectives}
For the isolated mode, the solution to the ED problem of an isolated BESS network mainly includes two parts. One is a MC consensus scheme, and the other is an average power mismatch estimation scheme. So, a distributed ED scheme with PI+R controllers, as shown in Fig. \ref{figpir}, is as follows,
\begin{subequations}
	\label{3}
	\begin{equation}
		\label{lambdaxi}
		\xi_{\lambda,i}^k=\sum_{i=1}^na_{ij}(\lambda_i^k-\lambda_j^k),
	\end{equation}
	\begin{equation}
		\label{res1}
		\nu_{\lambda,i}^{k}=\begin{cases}\nu_{\lambda,i}^{k-1}+\xi_{\lambda,i}^{k}, if \xi_{\lambda,i}^{k-1}\xi_{\lambda,i}^{k}>0,\\\xi_{i}^{k}, if \xi_{\lambda,i}^{k-1}\xi_{\lambda,i}^{k}\leq0,\end{cases}
	\end{equation}
	\begin{equation}
		\label{lambdai}
		\lambda_{i}^{k+1}=\lambda_{i}^{k}-h_{1}\xi_{\lambda,i}^{k}-h_{2}\nu_{\lambda,i}^{k}+\sigma^{k}\Delta\hat{P}_{a,i}^{k},
	\end{equation}
	\begin{equation}
		\label{xiP}
		\xi_{P,i}^k=\sum_{i=1}^na_{ij}(\Delta\hat{P}_{a,i}^k-\Delta\hat{P}_{a,j}^k),
	\end{equation}
	\begin{equation}
		\label{DPnu}
		\nu_{P,i}^{k}=\begin{cases}\nu_{P,i}^{k-1}+\xi_{P,i}^{k-1}, if \xi_{P,i}^{k-1}\xi_{P,i}^{k}>0,\\\xi_{P,i}^{k-1}, if \xi_{P,i}^{k-1}\xi_{P,i}^{k}\leq0,\end{cases}
	\end{equation}
	\begin{equation}
		\label{DPi}
		\Delta\hat{P}_{a,i}^{k+1}=\Delta\hat{P}_{a,i}^k-z_1\xi_{P,i}^k-z_2\nu_{P,i}^k+\Delta P_i^{k+1}-\Delta P_i^k,
	\end{equation}
\end{subequations}
where $\Delta\hat{P}_{a,i}^k$ is the average power mismatch estimated by agent $i$, $\sigma^{k}$ is a designed feedback gain, $h_1$, $h_2$, $z_1$, and $z_2$ are the gains to be designed, $\Delta P_i^k=P_{D,i}-P_i^k$. Then, a compact form of this scheme is
\begin{figure}
	\centering
	\includegraphics[width=7cm]{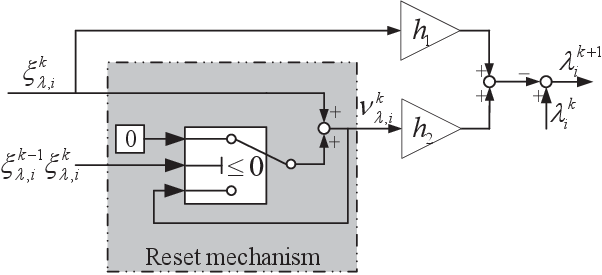} \caption{A PI controller with a error-dependent reset mechanism.\label{figpir}}
\end{figure}
\begin{subequations}
	\begin{equation}
		\xi_\lambda^k=L\lambda^k,
	\end{equation}
	\begin{equation}
		\label{lambdanu}
		\nu_\lambda^k=\begin{cases}\quad\nu_\lambda^{k-1}+\xi_\lambda^{k-1}, if k\in\mathcal{F}_\lambda,\\\nu_\lambda^{(k-1)+}+\xi_\lambda^{k-1}, if k\in\mathcal{J}_\lambda,\end{cases}
	\end{equation}
	\begin{equation}
		\label{lambda}
		\lambda^{k+1}=\lambda^k-h_1\xi_\lambda^k-h_2\nu_\lambda^k+\sigma^k\Delta\hat{P}_a^k,
	\end{equation}
	\begin{equation}
		\xi_P^k=L\Delta\hat{P}_a,
	\end{equation}
	\begin{equation}
		\label{Pnu}
		\nu_P^k=\left\{\begin{matrix}\nu_P^{k-1}+\xi_P^{k-1}, if k\in\mathcal{F}_P,\\\nu_P^{(k-1)+}+\xi_P^{k-1}, if k\in\mathcal{J}_P,\end{matrix}\right.
	\end{equation}
	\begin{equation}
		\label{pmf}
		\Delta\hat{P}_{a}^{k+1}=\Delta\hat{P}_{a}^{k}-z_{1}\xi_{P}^{k}-z_{2}\nu_{P}^{k}+\Delta P^{k+1}-\Delta P^{k},
	\end{equation}
\end{subequations}
where $\lambda^k,\xi_{\lambda}^{k},\nu_{\lambda}^{k},\Delta\hat{P}_{a}^{k},\xi_{P}^{k},\nu_{P}^{k}$, and $\Delta P^{k}$ are, respectively, stack vectors for $\lambda_i^{k},\xi_{\lambda,i}^{k},\nu_{\lambda,i}^{k},\Delta\hat{P}_{a,i}^{k},\xi_{P,i}^{k},\nu_{P,i}^{k}$, and $\Delta P_i^{k}$ for $\forall i\in \{ 1, 2, \cdots , n\} , \mathcal{J} _{\lambda }: \{ k| k\in l_{\lambda }\}$ ($\mathcal{J} _{P}: \{ k| k\in l_{P}\}$) and $\mathcal{F}_\lambda:\{k|k\notin l_\lambda\}$ ($\mathcal{F}_P:\{k|k\notin l_P\}$), respectively, are a jump set and a flow set, and $l_\lambda:\{k|\xi_{\lambda,i}^{k-1}\xi_{\lambda,i}^k\leq0,i\in$ $\{ 1, 2, \cdots , n\} \}$ ($l_{P}: \{ k| \xi _{P, i}^{k- 1}\xi _{P, i}^{k}\leq 0, i\in \{ 1, 2, \cdots , n\} \}$) is a jump instant set, $\nu_{\lambda}^{(k-1)+}$ ($\nu_{P}^{(k-1)+}$) is a stack vector obtained by resetting certain components of $\nu _{\lambda }^{( k- 1) }$ ($\nu _{P}^{( k- 1) }$) that meet the reset criteria to 0.
\subsection{Analysis of The Base System}
Considering capacity limitation, $\|\Delta P^{k+1}-\Delta P^k\|<+\infty $. Thus, the base system of \eqref{pmf} is
\begin{equation}
	\label{base1}
	x_P^{k+1}=\Phi(L)x_P^k,
\end{equation}
where $x_P^k=[(\xi_P^k)^T\quad(\nu_P^k)^T]^T$, $ \Phi(L)=\begin{bmatrix}I-z_1L&-z_2L\\I&I\end{bmatrix}$. Thus, the characteristic polynomial of $\Phi(L)$ is
$$\begin{aligned}|\mu I-\Phi(L)|=\prod_{i=1}^{n}|&\mu_{i}^{2}+(z_{1}\eta_{i}(L)-2)\mu_{i}+1\\
	&+\eta_{i}(L)(z_{2}-z_{1})|,\end{aligned}$$
which gives
\begin{equation}
	\mu_i=\frac{2-z_1\eta_i(L)\pm\sqrt{z_1^2\eta_i^2(L)-4z_2\eta_i(L)}}{2}.
\end{equation}
For \eqref{base1}, its properties are summarized in Lemma \ref{eig1} and its proof.
\begin{lemma}
	\label{eig1}
	Under Assumption \ref{assu1}, the base system \eqref{base1}, defined in \cite{banosResetControlSystems2011}, is asymptotically stable if $\frac{2(z_1+z_2)}{z_{1}^{2}}\ge\eta_M$ and $\frac{4z_2}{z_{1}^{2}}>\eta_m$, or $\frac{4z_{2}}{z_{1}^{2}}\le\eta_{m}$ with $z_{1} \leq \frac{2}{\eta_M}$ or $z_{1}>\frac{2}{\eta_M}$ and $2z_{1}-z_{2}<\frac{4}{\eta_M}$. Furthermore, the matrix $\Phi(L)$ has at least a pair of conjugate complex eigenvalues if $\frac{2(z_1+z_2)}{z_{1}^{2}}>\eta_M$ and $\frac{4z_2}{z_{1}^{2}}>\eta_m$,
	and the base system is so-called regular \cite{banosResetControlSystems2011}.
\end{lemma}
\begin{proof}
	We first consider the case where $\Delta_i<0$ for $i\in\{1,2,\cdots,n\}$, i.e., $\frac{4z_2}{z_{1}^{2}}>\eta_{m}$, where $\Delta_i=z_1^2\eta_i^2(L)-4z_2\eta_i(L)$. By calculation, it can be inferred that $|\mu_{i,j}|<1$ for $\forall j\in\{1,2\}$ if $\frac{2(z_1+z_2)}{z_{1}^{2}}\ge\eta_M$.
	\par Next, we consider a real eigenvalue caused by $\Delta_i\geq0$, that is, $\frac {4z_{2}}{z_{1}^{2}}\leq \eta_m$. If $2-z_{1}\eta _{i}(L)\ge 0$ and $2-z_{1}\eta_{i}(L)+\sqrt{\Delta_{i}}<2$, i.e., $z_{1} \leq \frac{2}{\eta_M}$, $|\mu_{i,j}|< 1$ for $j \in \{1,2\}$. Besides, if $2-z_{1}\eta_{i}(L)<0$ and $(2z_{1}-z_{2})\eta_{i}(L)<4$, i.e., $z_{1}>\frac{2}{\eta_M}$ and $2z_{1}-z_{2}<\frac{4}{\eta_M}$, $|\mu_{i,j}|< 1$ for $j \in \{1,2\}$.
	\par Thus, it can be asserted that $\Phi$ has at least a pair of conjugate
	complex eigenvalues if $\frac{4z_{2}}{z_{1}^{2}}>\eta_{m}$. Furthermore, under $\frac{2(z_1+z_2)}{z_{1}^{2}}\ge\eta_M$, the base system \eqref{base1} is asymptotically stable. Besides, if $\frac{4z_{2}}{z_{1}^{2}}\le\eta_{m}$ with $z_{1} \leq \frac{2}{\eta_M}$ or $z_{1}>\frac{2}{\eta_M}$ and $2z_{1}-z_{2}<\frac{4}{\eta_M}$, the base system \eqref{base1} is asymptotically stable. So far, Lemma \ref{eig1} can be derived.
\end{proof}
\par Based on Lemma \ref{eig1}, the base system of \eqref{lambda} with \eqref{lambdanu} can be represented as \eqref{base2},
\begin{equation}
	\label{base2}
	x_\lambda^{k+1}=\Phi(L)x_\lambda^k,
\end{equation}
where $x_\lambda^k=[(\xi_\lambda^k)^T\quad(\nu_\lambda^k)^T]^T$. Referring to Lemma \ref{eig1}, we give Lemma \ref{eig2} without proof.
\begin{lemma}
	\label{eig2}
	Under Assumption \ref{assu1}, if $\frac{2(h_1+h_2)}{h_{1}^{2}}\ge\eta_M$ and $\frac{4h_2}{h_{1}^{2}}>\eta_m$, the base system \eqref{base2} is asymptotically stable, i.e., regular \cite{banosResetControlSystems2011}.
\end{lemma}
\par Based on Lemmas \ref{eig1} and \ref{eig2}, we provide the following theorem
regarding the enabling conditions of the reset mechanisms in the designed PI+R controllers.
\begin{theorem}
	\label{TH1}
	Assume that the communication graph of the MAS governing isolated BESSs satisfies Assumption \ref{assu1}. Thus, the reset control system \eqref{base2} with \eqref{lambdanu} and $\Phi(L)$ is regular with at least one reset instant if $\frac{2(h_1+h_2)}{h_{1}^{2}}\ge\eta_M$ and $\frac{4h_2}{h_{1}^{2}}>\eta_m$. And the reset control system \eqref{base1} with \eqref{Pnu} and $\Phi(L)$ is regular with at least one reset instant if $\frac{2(z_1+z_2)}{z_{1}^{2}}\ge\eta_M$ and $\frac{4z_2}{z_{1}^{2}}>\eta_m$.
\end{theorem}
\begin{remark}
	Here, the gain condition in Theorem 1 is obtained by analyzing the eigenvalues of the base system \eqref{base2}. In contrast, in \cite{chenDistributedEconomicDispatch2021}, the gain conditions given by the Gale disk theorem are more conservative.
\end{remark}
\subsection{The Effectiveness of The Designed Scheme}
\par Define two error vectors $\delta_\lambda^k$ and $\delta_P^k$ as $\delta_\lambda^k=(I-J)\lambda^k$ and $\delta_P^k=\Delta\hat{P}_a^k-0_n$, respectively, where $J=\frac{1}{n}1_{n}1_{n}^{T}$. Then, referring \eqref{lambda}, \eqref{pmf}, and $1^T_nL=0_n$, we have
\begin{subequations}
	\begin{equation}
		\label{dlambda}
		\delta_{\lambda}^{k}=\delta_{\lambda}^{k-1}-h_{1}\xi_{\lambda}^{k-1}-h_{2}\nu_{\lambda}^{k-1}+\sigma^{k}(I-J)\Delta\hat{P}_{a}^{k},
	\end{equation}
	\begin{equation}
		\label{dP}
		\delta_P^k=\delta_P^{k-1}-h_1\xi_P^{k-1}-h_2\nu_P^{k-1}+\Delta P^{k+1}-\Delta P^k.
	\end{equation}
\end{subequations}
The compact form of \eqref{dlambda} and \eqref{dP} are written as
\begin{subequations}
	\begin{equation}
		\label{dlambda1}
		y_\lambda^{k+1}=\Phi(L)y_\lambda^k+\begin{bmatrix}\sigma^k(I-J)\Delta\hat P_a^k\\0\end{bmatrix},
	\end{equation}
	\begin{equation}
		\label{dP1}
		y_P^{k+1}=\Phi(L)y_P^k+\begin{bmatrix}\Delta P^{k+1}-\Delta P^k\\0\end{bmatrix},
	\end{equation}
\end{subequations}
where $y_{P}^{k}=col(\delta_{P}^{k},\nu_{P}^{k})$, $y_{P}^{k}=col(\delta_{\lambda}^{k},\nu_{\lambda}^{k})$. For \eqref{dlambda1} and \eqref{dP1}, convergence and stability are summarized in Theorems \ref{TH2} and
\ref{TH3}, and relevant proofs are provided, respectively.
\begin{theorem}
	\label{TH2}
	For the isolated mode, with the premise of Theorem \ref{TH1}, MC consensus with static errors can be guaranteed under a static gain $\sigma$. If a dynamic gain $\sigma^k>0$ for $k\in\{0,1,2,\cdots\}$ is selected, zero static error consensus of MC, i.e., $\lim\limits_{k\to+\infty}|\lambda_{i}-\lambda_{j}| = 0$ for $i,j\in\{1,2,\cdots,n\}$, can be guaranteed, where $\lim\limits_{k\to+\infty}\sigma^{k}=0$ and $\lim\limits_{k\to+\infty}\sum_{j=0}^{k}\sigma^{j}=+\infty $. At the same time, MC of each BESS is asymptotically stable. In addition, the estimated average power mismatch is bounded, i.e., $\lim\limits_{k\to+\infty}|\Delta\hat{P}_{a,i}|<+\infty$ for $i\in\{1,2,\cdots,n\}$.
\end{theorem}
\begin{proof}
	Obviously, at any reset instant $l_{P,j}\in l_{P}$ and $l_{\lambda,j}\in l_{\lambda}$, there must be
	$$\|y_P^{l_{P,j}+}\|\leq\|y_P^{l_{P,j}}\|\,\text{and}\,\|y_\lambda^{l_{\lambda,j}+}\|\leq\|y_\lambda^{l_{\lambda,j}}\|.$$
	Besides, $\|\Delta P^k\|$ for $\forall k \in \{1,2,\cdots\}$ is bounded according to \eqref{power}. Thus, it is easy to verify that each component of $\|(I-J)(\Delta P^{k+1}-\Delta P^{k})\|$ is bounded by a
	positive constant denoted as $\Delta\bar{P}$ with $0\leq\Delta\bar{P}<+\infty $. Before the first reset instant,
	\begin{equation}
		\label{10}
		y_P^k=\Phi^k(L)y_P^0+\sum_{j=0}^{k-1}\Phi^{k-1-j}(L)\begin{bmatrix}\Delta P^{j+1}-\Delta P^j\\0\end{bmatrix}.
	\end{equation}
	For the first term on the right side of \eqref{10}, we have $\|\Phi^{k}(L)y_{P}^{0}\|\leq\|\Phi^{k}(L)\|\|y_{P}^{0}\|\leq\alpha^{k}\|y_{P}^{0}\|$, where $\alpha$ is the spectral radius of the matrix $\Phi(L)$ and $0<\alpha<1$ according to Theorem \ref{TH1}. Hence, one can get
	\begin{equation}
		\begin{aligned}
			&\|\sum_{j=0}^{k-1}\Phi^{k-1-j}(L)\begin{bmatrix}\Delta P^{j+1}-\Delta P^{j}\\0\end{bmatrix}\|\\
			\leq&\sum_{j=0}^{k-1}\|\Phi^{k-1-j}(L)\|\|\left\lfloor\begin{matrix}\Delta P^{j+1}-\Delta P^{j}\\0\end{matrix}\right\rfloor\|
			\leq\frac{1-\alpha^{k-1}}{1-\alpha}\Delta\bar{P}.
		\end{aligned}
	\end{equation}
	Anyway, $\|y_{P}^{k}\|\leq\alpha^{k}\|y_{P}^{0}\|+\frac{1-\alpha^{k-1}}{1-\alpha}\Delta\bar{P}$. For $l_{P,1}<k\leq l_{P,2}$, we have
	\begin{equation}
		y_P^k=\Phi^{k-l_{P,1}}(L)y_P^{l_{P,1}+}+\sum_{j=l_{P,1}}^{k-1}\Phi^{k-1-j}(L)\begin{bmatrix}\Delta P^{j+1}-\Delta P^j\\0\end{bmatrix}.
	\end{equation}
	Thus, there must be
	$$\begin{aligned}
		\|y_{P}^{k}\|&\leq\|\Phi^{k}(L)y_{P}^{l_{P,1}+}\| 
		+\|\sum_{j=l_{1}}^{k-1}\Phi^{k-1-j}(L)\left[\begin{matrix}{\Delta P^{j+1}-\Delta P^{j}}\\{0}\\\end{matrix}\right]\| \\
		&\leq\alpha^{k-l_{P,1}}\|y_{P}^{l_{P,1}+}\|+\frac{\alpha^{l_{1}}-\alpha^{k-1}}{1-\alpha}\Delta\bar{P} \\
		&\leq\alpha^{k}\|y_{P}^{0}\|+\frac{\alpha^{k-l_{1}}-\alpha^{k-1}}{1-\alpha}\Delta\bar{P}+\frac{1-\alpha^{k-l_{1}-1}}{1-\alpha}\Delta\bar{P} \\
		&=\alpha^{k}\|y_{P}^{0}\|+\frac{1-\alpha^{k-1}}{1-\alpha}\Delta\bar{P}.
	\end{aligned}$$
	According to the Mathematical Induction method, for any moment, it conforms to the above equation for $\forall k\in\{1,2,\cdots\}$. Then,
	\begin{equation}
		\label{DP}
		\|y_P^k\|\leq\frac{1}{1-\alpha}\Delta\bar{P}, k\to+\infty,
	\end{equation}
	which implies $\delta_{P,i}^k$, i.e., $\Delta\hat{P}_{a,i}^{k}$ and $\nu_{P,i}^k$,
	are bounded for $i\in\{1,2,\cdots,n\}$.
	\par Next, \eqref{dlambda1} can be deduced as
	\begin{equation}
		y_\lambda^k=\Phi^k(L)y_\lambda^0+\sum_{j=0}^{k-1}\Phi^{k-1-j}(L)\begin{bmatrix}(I-J)\sigma\Delta\hat{P}_a^k\\0\end{bmatrix}.
	\end{equation}
	Based on the previous results, two conclusions on MC are directly given here,
	$$\lim_{k\to+\infty}\|\Phi^{k}(L)y_{\lambda}^{0}\|=0.$$
	$$\begin{aligned}\|\sum_{j=0}^{k-1}\Phi^{k-1-j}(L)\begin{bmatrix}\sigma\Delta\hat{P}_a^k\\0\end{bmatrix}\|&\leq\sum_{j=0}^{k-1}\|\Phi^{k-1-j}(L)\sigma\|\|\begin{bmatrix}\Delta\hat{P}_a^k\\0\end{bmatrix}\|\\&<+\infty.\end{aligned}$$
	Thus, $\lim\limits_{k\to+\infty}\|y_{\lambda}^{k}\|<+\infty$.
	\par When introducing a time-varying gain $\sigma^k$ as shown in Theorem \ref{TH2}, we have
	\begin{equation}
		\label{lam1}
		\begin{aligned}
			\parallel\sum_{j=0}^{k-1}\Phi^{k-1-j}\begin{bmatrix}\sigma^{k}\Delta\hat{P}_{a}^{k}\\0\end{bmatrix}\parallel 
			\leq&\sum_{j=0}^{k-1}\|\Phi^{k-1-j}\sigma^{k}\|\|\begin{bmatrix}\Delta\hat{P}_{a}^{k}\\0\end{bmatrix}\| \\
			<&o(1)\sum_{j=0}^{k-1}\|\begin{bmatrix}\Delta\hat{P}_{a}^{k}\\0\end{bmatrix}\|=o(1).
		\end{aligned}
	\end{equation}
	In this way, $\|y_\lambda^k\|=o(1)$ with $k\to+\infty$, which leads to $\lim\limits_{k\to+\infty}\|\delta_{\lambda}^{k}\|=0$. Thus, consensus on MC can be reached and $\lim\limits_{k\to+\infty}\xi_{\lambda,i}^{k}=0$ for $\forall i\in\{1,2,\cdots,n\}$.
	\par Besides, $\|y_\lambda^k\| = o(1)$ with $k\to+\infty$ gives $\lim\limits_{k\to+\infty}\|(I-J)\nu_\lambda^k\| = 0$, which directly leads to $\lim\limits_{k\to+\infty}\nu_{\lambda,i}^{k} = 0$ for $\forall i \in \{1,2,\cdots,n\}$. Hence, based on the above conclusions and \eqref{lambdai}, we have
	$$\lim_{k\to+\infty}(\lambda_i^{k+1}-\lambda_i^k)=\lim_{k\to+\infty}\sigma^k\Delta\hat{P}_a^k,$$
	for $\forall i \in \{1,2,\cdots,n\}$, which is further deduced as
	$$\lim_{k\to+\infty}(\lambda_i^{k+1}-\lambda_i^k)=0,$$
	for $\forall i \in \{1,2,\cdots,n\}$ according to \eqref{DP}. So, it can be asserted that MC for each BESS is asymptotically stable.
\end{proof}
\begin{remark}
	Here, take MC as an example to briefly describe the progressiveness of the involved scheme. Compared to the proportional protocol in \cite{chenDistributedEconomicDispatch2021}, the designed scheme can accelerate the convergence rate of MC. According to the reset
	mechanism \eqref{res1}, the sign of $\nu_{\lambda,i}^k$ is aligned with $\xi_{\lambda,i}^k$ at all instants. Thus,
	$$h_1\xi_{\lambda,i}^k+h_2\nu_{\lambda,i}^k\geq h_1\xi_{\lambda,i}^k.$$
	As a result, the convergence rate of $\lambda_i$ is accelerated. It can
	be seen that the introduction of reset mechanism enhances the consensus of the system, which is further compared in the case study.
\end{remark}
\begin{remark}
	Here, a discussion is organized to explain the impact of gains on the designed schemes, taking the MC consensus controller as an example. If $\frac{2(h_1+h_2)}{h_{1}^{2}}\ge\eta_M$ and $\frac{4h_2}{h_{1}^{2}}>\eta_m$, $\vert\mu_{ij}\vert=\sqrt{1-((h_1+h_2)\eta_i(L)-\frac{1}{2}h_1^2\eta_i^2(L))}$. Thus, $\vert\mu_{ij}\vert$ decreases under a increasing $h_2$. The impact of $h_1$ on $\vert\mu_{ij}\vert$ is relatively complex. If $h_1<\frac{1}{\eta_M}$, $\vert\mu_{ij}\vert$ increases with the increase of $h_1$; if $h_1>\frac{1}{\eta_m}$, the opposite is true. If $\frac{4h_{2}}{h_{1}^{2}}\le\eta_{m}$ and $h_{1} \leq \frac{2}{\eta_M}$, $\vert\mu_{ij}\vert$ increases under increasing $h_2$ or/and $h_1$. If $\frac{4h_{2}}{h_{1}^{2}}\le\eta_{m}$, $h_{1}>\frac{2}{\eta_M}$, and $2h_{1}-h_{2}<\frac{4}{\eta_M}$, $\vert\mu_{ij}\vert$ increases under decreasing $h_2$ or/and $h_1$. Although the situation analyzed above is somewhat conservative, the impact of its complement on the results is very unclear and may require trial and error to explore.
\end{remark}
\par Below is a theorem to illustrate the stability of MC and the average power mismatch estimation in the scheme \eqref{3}.
\begin{theorem}
	\label{re3}
	Based on the result of Theorem \ref{TH2}, the stability of MC and the average power mismatch estimation can be achieved.
\end{theorem}
\begin{proof}
	Due to the strong nonlinear relationship between feedback term $\Delta P_{i}^{k+1}-\Delta P_{i}^{k}$ and $\lambda_i^k$, it is difficult to accurately characterize the stability of the involved scheme. But there are still traces to follow. Multiplying \eqref{lambda} and \eqref{pmf} left by $1^T$ yields
	$$\begin{bmatrix}1^T\lambda^{k+1}\\1^T\Delta\hat P_a^{k+1}\end{bmatrix}=\begin{bmatrix}1&\sigma^k\\0&1\end{bmatrix}\begin{bmatrix}1^T\lambda^k\\1^T\Delta\hat P_a^k\end{bmatrix}+\begin{bmatrix}0\\1^T(\Delta P^{k+1}-\Delta P^k)\end{bmatrix}.$$
	We note that within the limits, power is positively correlated with MC according to \eqref{power}. In view of this, a bounded time-varying positive parameter $f_i^k$ is introduced such that
	$$\begin{aligned}
		\Delta P_{i}^{k+1}-\Delta P_{i}^{k} =&(B_i(P_i^{k+1}+P_i^k)-1)(P_i^{k+1}-P_i^k)\\
		=&-f_i^k(\lambda_i^{k+1}-\lambda_i^k). 
	\end{aligned}$$
	At any step, there must be an $f_m^k$ and $f_M^k$ that
	$$1^{T}f_{m}^{k}(\lambda^{k+1}-\lambda^{k})\leq1^{T}f^{k}(\lambda^{k+1}-\lambda^{k})\leq1^{T}f_{M}^{k}(\lambda^{k+1}-\lambda^{k}).$$
	Noting that $1^{T}(\lambda^{k+1}-\lambda^{k})=1^{T}\sigma^{k}\Delta\hat{P}_{a}$, construct the following two systems,
	$$\begin{bmatrix}1^T\lambda^{k+1}\\1^T\Delta\hat{P}_a^{k+1}\end{bmatrix}=\begin{bmatrix}1&\sigma^k\\0&1-\sigma^kf_M^k\end{bmatrix}\begin{bmatrix}1^T\lambda^k\\1^T\Delta\hat{P}_a^k\end{bmatrix},$$
	$$\begin{bmatrix}1^T\lambda^{k+1}\\1^T\Delta\hat{P}_a^{k+1}\end{bmatrix}=\begin{bmatrix}1&\sigma^k\\0&1-\sigma^kf_m^k\end{bmatrix}\begin{bmatrix}1^T\lambda^k\\1^T\Delta\hat{P}_a^k\end{bmatrix}.$$
	Since $f_m^k$ and $f_M^k$ are bounded and $\sigma^k$ is decaying, we can
	assert that from a certain step on, $0\le\sigma^kf_M^k\le 2$ and $0\le\sigma^kf_m^k\le 2$. That is, $\begin{bmatrix}1&\sigma^k\\0&1-\sigma^kf_M^k\end{bmatrix}$ and $\begin{bmatrix}1&\sigma^k\\0&1-\sigma^kf_m^k\end{bmatrix}$ will not have eigenvalues with a modulus greater than 1. These indicate that under the designed scheme \eqref{3}, MC and the estimated average power mismatch of each BESS are stable.
\end{proof}
\begin{remark}
	The validity of Theorem \ref{re3} depends on the coefficients of the cost function. In this article, since the cost function coefficients of a battery are all positive, the scheme proposed in \eqref{3} in this article can lead to stability. For some power sources with negative cost coefficients, more discussion is needed to analyze stability.
\end{remark}
\par According to the result of Theorem \ref{re3}, MC of each BESS is stable. Then, the stability of the estimated average power mismatch can be further derived in the following theorem.
\begin{theorem}
	\label{TH3}
	Based on Theorem \ref{re3}, the estimated power mismatch asymptotically converges to 0 under the average power mismatch estimation scheme \eqref{DPi} with \eqref{DPnu} and \eqref{xiP}.
\end{theorem}
\begin{proof}
	Due to the stability of MC, the output power $P_i^k$ for $\forall i\in\{1,2,\cdots,n\}$ is convergent calculated by \eqref{power}. Thus, according to the definition of $\Delta P_i^k$ for $\forall i\in\{1,2,\cdots,n\}$ is convergent, which implies $\Delta\bar{P}_{i}=o(1)$ with $k\to+\infty$. Based on this and \eqref{DP}, it can only be inferred that $\lim\limits_{k\to+\infty}\|y_P^k\|=0$. Further, we arrive at
	\begin{equation}
		\lim_{k\to+\infty}\|\delta_P^k\|=0,
	\end{equation}
	which leads that $\Delta{\hat P}_{a,i}^k$ for $\forall i\in\{1,2,\cdots,n\}$ converges to 0 asymptotically. Thus, we can conclude that Theorem \ref{TH3} holds.
\end{proof}
\begin{remark}
	Theorems \ref{re3} and \ref{TH3} perfectly fill some of the regrets in \cite{chenDistributedEconomicDispatch2021}, where the estimated average power mismatch has not been proven to converge to 0. However, in this article, the boundedness and convergence of $\Delta{P}_i^{k+1}-\Delta{P}_i^k$ are cleverly applied to prove that the estimated average power mismatch asymptotically converges to 0.
\end{remark}
\par So, there is a following proposition about the ED problem of an isolated BESS network.
\begin{proposition}
	According to the result of Theorems \ref{TH2} and \ref{TH3}, the ED problem of an isolated BESS network can be solved by using \eqref{power}, \eqref{lambdai} with \eqref{res1} and \eqref{lambdaxi}, and \eqref{DPi} with \eqref{xiP} and
	\eqref{DPnu}.
\end{proposition}
\begin{proof}
	The consensus and asymptotic stability on MC have been proven in Theorem \ref{TH2}. In Theorem \ref{TH3}, the stability of the estimated average power mismatch has been proven. Next, the supply-demand balance is needed to be explained under the current situation.
	\par With the help of $1^TL=0$ and $\nu_P^l=0$, the following derivation holds true,
	\begin{equation}
		\label{balance}
		\begin{aligned}
			&1^T[\Delta\hat{P}_{a}^{k+1}-\Delta P^{k+1}]\\
			=&1^T[(I-h_1L)\Delta\hat{P}_a^k-h_2v_P^k-\Delta P^k] \\
			=&1^{T}[\Delta\hat{P}_{a}^{k}-h_{2}v_{P}^{k}-\Delta P^{k}]
			=\cdots\\
			=&1^T[\Delta\hat{P}_a^0-h_2\sum_{j=0}^kv_P^j-\Delta P^0] \\
			=&1^{T}[\Delta\hat{P}_{a}^{0}-\Delta P^{0}]-h_{2}1^{T}\sum_{j=0}^{k}v_{P}^{j} \\
			=&-h_21^T\sum_{l=1}^{R^{k+1}}v_P^l=0.
		\end{aligned}
	\end{equation}
	According to the result of Theorem \ref{TH3}, it can be concluded that
	$$\lim_{k\to+\infty}1^T\Delta P^k=0.$$
	According to the definition of power mismatch, it can be asserted that the supply-demand balance can be asymptotically maintained.
	\par At this point, the ED problem regarding an isolated BESS network has been resolved.
\end{proof}
\begin{remark}
	A decaying feedback term $\sigma^k\Delta\hat{P}^k$ is introduced into the MC control scheme, which does not affect the consensus of MCs, but does not lead to average consensus. Because of this, the total output of all BESSs under this scheme with a feedback term $\sigma^k\Delta\hat{P}^k$ can afford the load and dynamic line loss.
\end{remark}
\subsection{Further Discussion}
In this section, a distributed ED scheme with PI+R controllers is redesigned by introducing the feedback item $\Delta\hat{P}^k$ into the closed-loop error $\xi_{\lambda,i}$. Denote $y_{i}^{k}=\begin{bmatrix}\lambda_{i}^{k}\\\Delta\hat{P}_{a,i}^{k}\end{bmatrix}$. Redefine the difference item of MC is
\begin{equation}
	\label{xilambda}
	\xi_{\lambda,i}^k=\sum_{i=1}^na_{ij}(\lambda_i^k-\lambda_j^k)+\Delta\hat{P}_{a,i}.
\end{equation}
Denote $\xi_i^k=\begin{bmatrix}\xi_{\lambda,i}^k\\\xi_{P,i}^k\end{bmatrix}$ and $\nu_i=\begin{bmatrix}\nu_{\lambda,i}^k\\\nu_{P,i}^k\end{bmatrix}$. In a compact form, we have
\begin{subequations}
	\label{19}
	\begin{equation}
		y^{k+1}=y^k-h_1\xi^k-h_2\nu^k+(\Delta P^{k+1}-\Delta P^k)\otimes\begin{bmatrix}0\\1\end{bmatrix},
	\end{equation}
	\begin{equation}
		\left.\nu_{\lambda}^{k+1}=\left\{\begin{matrix}\nu_{\lambda}^{k}+\xi_{\lambda}^{k}, if\,k\in\mathcal{F}_{\lambda},\\\nu_{\lambda}^{k+}+\xi_{\lambda}^{k}, if\,k\in\mathcal{J}_{\lambda},\end{matrix}\right.\right.
	\end{equation}
	\begin{equation}
		\nu_P^{k+1}=\left\{\begin{matrix}\nu_P^k+\xi_P^k, if \,k\in\mathcal{F}_P,\\\xi_P^k,if\,k\in\mathcal{J}_P,\end{matrix}\right.
	\end{equation}
\end{subequations}
where $\otimes$ denote Kronecker product, and whose base system is
\begin{equation}
	\label{base3}
	\begin{bmatrix}\xi^{k+1}\\\nu^{k+1}\end{bmatrix}=\Phi(\phi(L))\begin{bmatrix}\xi^k\\\nu^k\end{bmatrix}
\end{equation}
with $\phi(L)=\begin{bmatrix}L&I\\0&L\end{bmatrix}$. The eigenvalues of $\phi(L)$ are the same as those of $L$, which leads that any eigenvalue of $\Phi(\phi(L))$ is
one of those of $\Phi(L)$. So, a lemma on the base system \eqref{base3} is given without the need for proof.
\begin{lemma}
	\label{lem3}
	Under Assumption \ref{assu1}, if $\frac{2(h_1+h_2)}{h_{1}^{2}}\ge\eta_M$ and $\frac{4h_2}{h_{1}^{2}}>\eta_m$, or $\frac{4h_{2}}{h_{1}^{2}}\le\eta_{m}$ with $h_{1} \leq \frac{2}{\eta_M}$ or $h_{1}>\frac{2}{\eta_M}$ and $2h_{1}-h_{2}<\frac{4}{\eta_M}$, the base system \eqref{base3} is
	asymptotically stable, i.e., regular.
\end{lemma}
\par It can be seen that the range of gains tuning of the base system will not be affected by introducing the feedback term into the closed-loop error. Next, the stability and convergence of the designed scheme \eqref{19} are explained in Theorem \ref{TH4} and its proof.
\begin{theorem}
	\label{TH4}
	According to Lemma \ref{lem3}, for an isolated BESS network, if $\frac{2(h_1+h_2)}{h_{1}^{2}}\ge\eta_M$ and $\frac{4h_2}{h_{1}^{2}}>\eta_m$, or $\frac{4h_{2}}{h_{1}^{2}}\le\eta_{m}$ with $h_{1} \leq \frac{2}{\eta_M}$ or $h_{1}>\frac{2}{\eta_M}$ and $2h_{1}-h_{2}<\frac{4}{\eta_M}$, MC and the estimated average power mismatch are bounded under the control scheme \eqref{19}.
\end{theorem}
\begin{proof}
	Referring \eqref{19}, we have
	\begin{subequations}
		\begin{equation}
			\label{dlambda2}
			\delta_\lambda^k=\delta_\lambda^{k-1}-h_1\xi_\lambda^{k-1}-h_2\nu_\lambda^{k-1}-J\Delta\hat{P}_a^k,
		\end{equation}
		\begin{equation}
			\label{dP2}
			\delta_P^k=\delta_P^{k-1}-h_1\xi_P^{k-1}-h_2\nu_P^{k-1}+\Delta P^{k+1}-\Delta P^k.
		\end{equation}
	\end{subequations}
	The compact form of \eqref{dlambda2} and \eqref{dP2} are written as
	\begin{subequations}
		\begin{equation}
			y_\lambda^{k+1}=\Phi(L)y_\lambda^k-\begin{bmatrix}J\Delta\hat{P}_a^k\\0\end{bmatrix},
		\end{equation}
		\begin{equation}
			y_P^{k+1}=\Phi(L)y_P^k+\begin{bmatrix}\Delta P^{k+1}-\Delta P^k\\0\end{bmatrix},
		\end{equation}
	\end{subequations}
	where $y_\lambda^k=\begin{bmatrix}\delta_\lambda^k\\\nu_\lambda^k\end{bmatrix}$ and $y_P^k=\begin{bmatrix}\delta_P^k\\\nu_P^k\end{bmatrix}$. According to \eqref{DP}, it can be inferred that
	$$\|y_P^k\|\leq\frac{1}{1-\alpha}\Delta\bar{P}, k\to+\infty.$$
	Similarly, in view of $1_n^T\Delta\hat{P}_a^k=1_n^T\Delta P^k$ according to \eqref{balance},
	$$\begin{aligned}
		\|y_{\lambda}^{k}\|\leq&\|\sum_{j=0}^{k-1}\Phi_{d}^{k-1-j}\begin{bmatrix}J\Delta\hat{P}_{a}^{k}\\0\end{bmatrix}\| 
		\leq\sum_{j=0}^{k-1}\|\Phi_{d}^{k-1-j}\|\|\begin{bmatrix}J\Delta\hat{P}_{a}^{k}\\0\end{bmatrix}\| \\
		\leq&\frac{1}{1-\alpha}\Delta\bar{P},k\to+\infty.
	\end{aligned}$$
	So far, Theorem \ref{TH4} has been proved.
\end{proof}
\begin{remark}
	At present, it can only be proven that MC and the estimated average power mismatch are bounded. Unfortunately, the stability proofs are still pending, which is attributed to the nonlinear relationship between $\Delta P_i^k$ and $\lambda_i$. However, following the approach in \cite{chenDistributedEconomicDispatch2021}, if one is stable, the other is naturally stable. Furthermore, if a time-varying gain $\sigma^k$ is introduced, the problem can be easily solved with the help of $\Delta {\hat P}_i^k$ being bounded, which is further explained in the following theorem. However, this may compromise the dynamic performance of this scheme.
\end{remark}
\par To ensure stability, a decay rate $\sigma^k$ is introduced to improve \eqref{xilambda}, i.e.,
\begin{equation}
	\label{xil}
	\xi_{\lambda,i}^k=\sum_{i=1}^na_{ij}(\lambda_i^k-\lambda_j^k)+\sigma^k\Delta\hat{P}_{a,i}^k.
\end{equation}
Next, a theorem and its proof are presented to demonstrate convergence.
\begin{theorem}
	\label{TH5}
	Based on Theorem \ref{TH4}, the ED problem can be solved by \eqref{19} with \eqref{xil}. That is, $\lambda_{i}^{k}-\lambda_{j}^{k}=0$ and $\Delta\hat{P}_{i}^{k} = 0, k \to +\infty $ for $i,j\in\{1,2,\cdots,n\}$.
\end{theorem}
\begin{proof}
	Under \eqref{xil}, one can get
	\begin{equation}
		\delta_{\lambda}^{k}=\delta_{\lambda}^{k-1}-h_{1}\xi_{\lambda}^{k-1}-h_{2}\nu_{\lambda}^{k-1}-J\sigma^{k}\Delta\hat{P}_{a}^{k},
	\end{equation}
	which is further derived as
	\begin{equation}
		y_{\lambda}^{k+1}=\Phi(L)y_{\lambda}^{k}-\begin{bmatrix}J\sigma^{k}\Delta\hat{P}_{a}^{k}\\0\end{bmatrix}.
	\end{equation}
	According to \eqref{lam1}, it can be obtained that
	$$\|y_\lambda^k\|=0, k\to+\infty,$$
	which implies that $\lambda_{i}-\lambda_{j}=0$, $k\to+\infty$ for $i,j \in \{1,2,\cdots,n\}$. Referring to \eqref{DP}, it can be asserted that
	$$\|y_P^k\|=0, k\to+\infty.$$
	So far, Theorem \ref{TH5} is proven.
\end{proof}
\begin{remark}
	Here, we briefly describe the progressiveness of scheme \eqref{3} compared with scheme \eqref{power}. In fact, due to the strong nonlinear term $\Delta P^{k+1}-\Delta P^{k}$, it is difficult to characterize the difference in convergence rates between the two schemes. We can compare the upper bound of states to discover some clues. Compared with \eqref{lambdaxi}, \eqref{xil} behaves with
	$$\begin{aligned}
		\|y_{\lambda}^{k+1}\|\leq &\|\Phi(L)\|^{k+1}\|y_{\lambda}^{0}\|
		+\|\sum_{j=0}^{k-1}\Phi^{k-1-j}(L)\begin{bmatrix}J\sigma^j\Delta\hat{P}_a^j\\0\end{bmatrix}\| \\
		\leq&\|\Phi(L)\|^{k+1}\|y_\lambda^0\|\\
		&+\|\sum_{j=0}^{k-1}\Phi^{k-1-j}(L)\begin{bmatrix}(I-J)\sigma^j\Delta\hat{P}_a^j\\0\end{bmatrix}\|,
	\end{aligned}$$
	which indicates that to a certain extent, the convergence rate of \eqref{xil} will not be lower than that of \eqref{lambdaxi}. Meanwhile, MC is smoother under scheme \eqref{xil}, which will be verified through subsequent simulations.
\end{remark}
\begin{remark}
	MC and estimated average power mismatch driven by the scheme \eqref{19} still conform to the conclusion of Theorem \ref{re3}, which is stability. Furthermore, as can be seen from Theorem \ref{re3}, this approach can accelerate consensus without adversely affecting stability.
\end{remark}
\section{Some Simulation Cases}
\begin{table}[h]
	\centering
	\caption{Comparison of Three Schemes}\label{table}
	\begin{tabular}{@{}cccc@{}}
		\toprule
		\makecell{Comparative\\Indicators}
		&\textbf{\eqref{3}}&{\eqref{19}}&\makecell{The scheme\\in \cite{chenDistributedEconomicDispatch2021}}\\
		\midrule
		Settling Time &$25s$&$20s$&$45s$\\
		\addlinespace[0.1cm]
		Consensus Time &$10s$&$10s$&$25s$\\
		\addlinespace[0.1cm]
		\makecell{BESS-oriented\\Play-and-Plug} &$>40s$&$20s$&$>40s$\\
		\addlinespace[0.1cm]
		\makecell{Agent-oriented\\Play-and-Plug} &$125s$&$10s$&$>175s$\\
		\addlinespace[0.1cm]
		Load Change &$>100s$&$20s$&$>100s$\\
		\addlinespace[0.1cm]
		\makecell{Large-scale\\Power Sytem} &$>150s$&$70s$&$>150s$\\
		\toprule
	\end{tabular}
\end{table}
\par In order to verify the effectiveness and progressiveness of the designed schemes, five cases are arranged. The simulation results are compared with the scheme designed in \cite{chenDistributedEconomicDispatch2021}. The adopted BESS network and its accompanying communication network are shown in Fig. \ref{fig3}. 
\par In Case 1, the designed schemes and the one in \cite{chenDistributedEconomicDispatch2021} are tested to compare settling time and consensus time.
\par In Cases 2 and 3, BESSs in the physical system and agents in the information system are tested and compared for plug and play, respectively.
\par In Cases 4 and 5, three schemes are test and compared under load switching and a large-scale system, respectively.
\par In order to more clearly describe the progressiveness, the relevant indicators in the simulation results in this section are summarized in Table \ref{table} to provide an intuitive comparison, where the latter four indicators uses the approximate time required to achieve supply-demand balance. Please refer to the simulation results and discussion in the corresponding case for other differences.
\begin{figure}
	\centering
	\includegraphics[width=8cm]{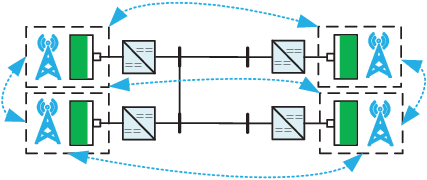} \caption{A BESS network with its communication graph.\label{fig3}}
\end{figure}
\begin{figure}
	\centering
	\includegraphics[width=8cm]{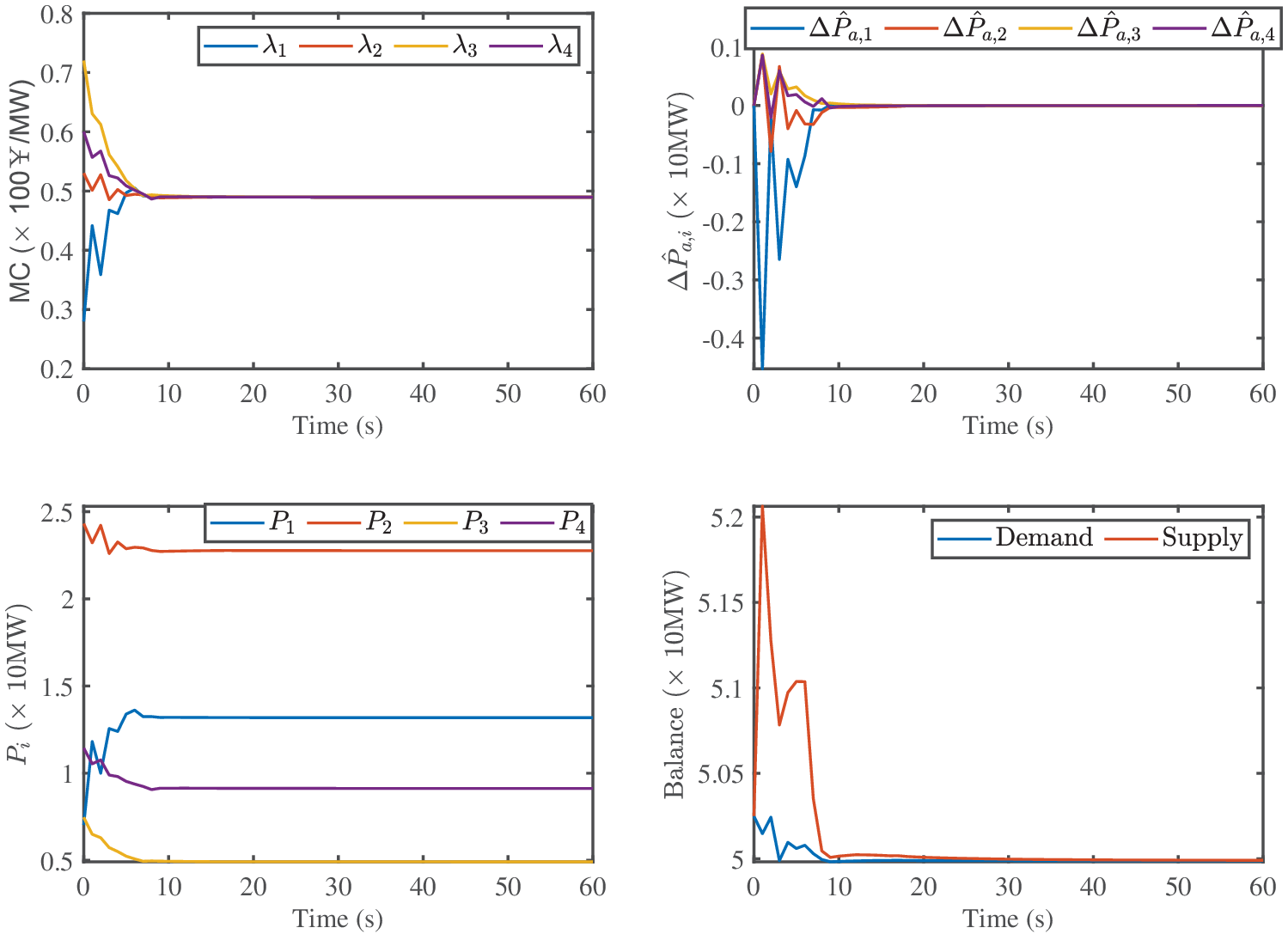} \caption{Simulation results of the scheme \eqref{3} in Case 1.\label{figc1a}}
\end{figure}
\begin{figure}
	\centering
	\includegraphics[width=8cm]{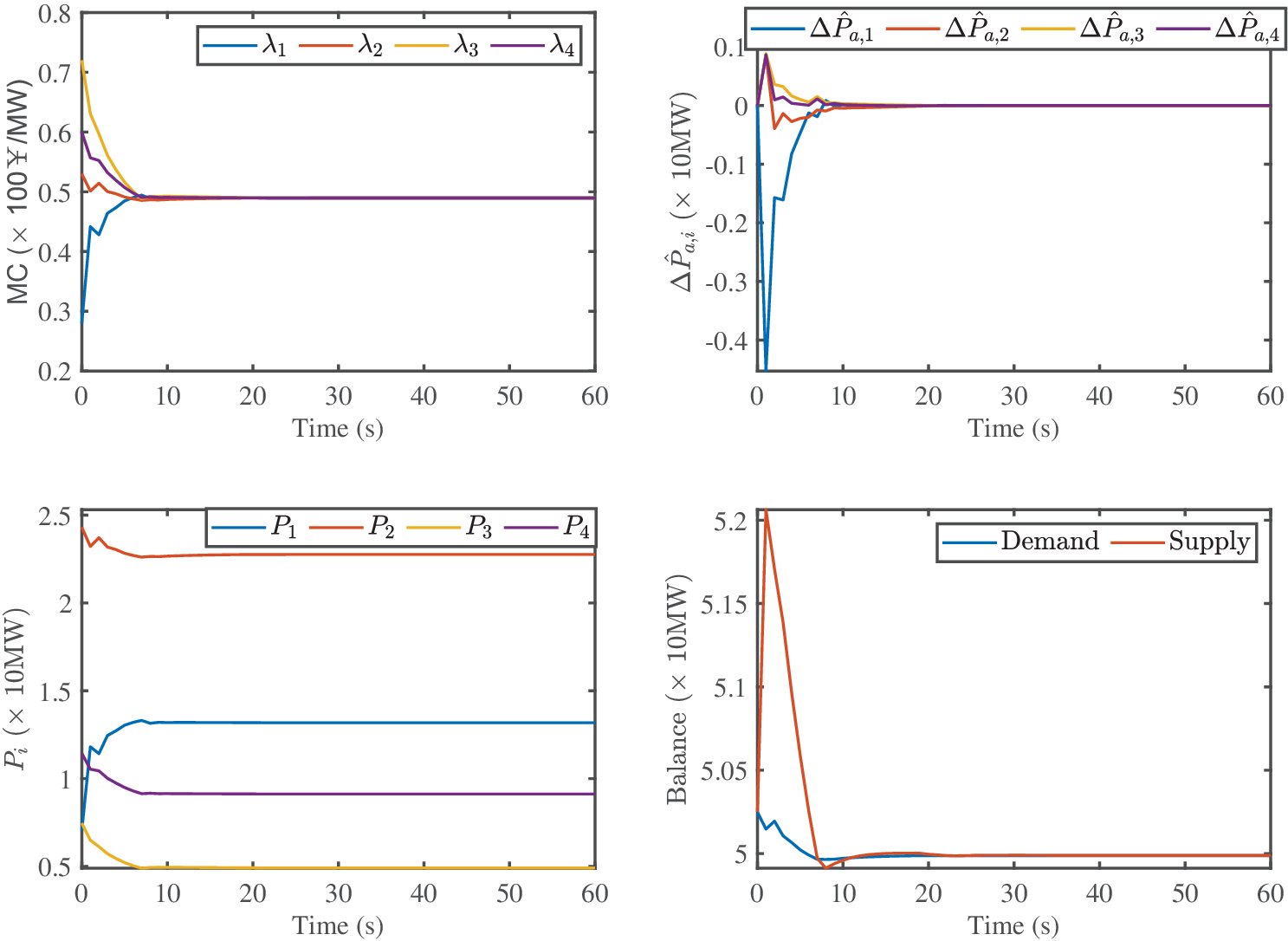} \caption{Simulation results of the scheme \eqref{19} in Case 1.\label{figc1b}}
\end{figure}
\begin{figure}
	\centering
	\includegraphics[width=8cm]{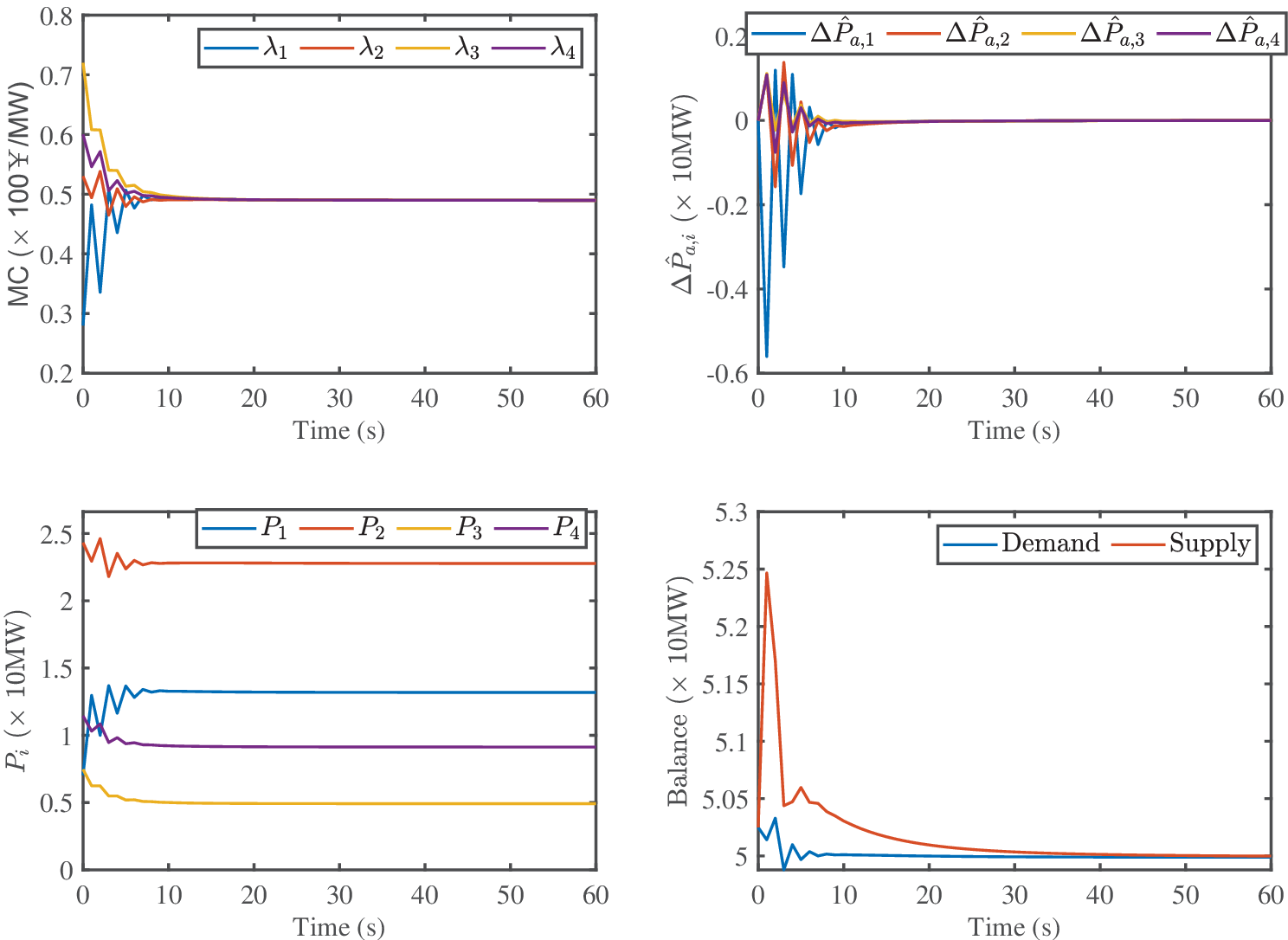} \caption{Simulation results of a distributed scheme designed in \cite{chenDistributedEconomicDispatch2021} in Case 1.\label{figc1c}}
\end{figure}
\subsection{Case 1. The Test on The Effectiveness and Progressiveness}
\par This case is arranged to investigate the designed isolated ED schemes. Here, the designed schemes \eqref{3} and \eqref{19} are tested under the same gains and compared with the scheme in \cite{chenDistributedEconomicDispatch2021}. The simulation results of \eqref{3} , \eqref{19}, and the comparison are shown in Figs. \ref{figc1a}-\ref{figc1c}, respectively.
\par From the simulation results in Figs. \ref{figc1a} and \ref{figc1b}, respectively, it can be seen that MCs can ultimately achieve consensus, which in turn leads to a stable output power sequence. Then, the estimated average power mismatch can converge to 0, and the supply-demand balance can be well maintained. However, there are still some obvious differences. In terms of convergence rate, scheme \eqref{19} performs better than scheme \eqref{3} . Specifically, the settling time of MC and the estimated average power mismatch, under the schemes \eqref{3} and \eqref{19}, are about $25s$ and $20s$, respectively. Under the scheme \eqref{3}, MC and the estimated average power mismatch appear more oscillatory, which result in a more significant fluctuation in the total supply power.
\par Previously, some researchers proposed a distributed ED scheme with proportional protocols in \cite{chenDistributedEconomicDispatch2021}, as shown below,
$$\lambda_i^{k+1}=\lambda_i^k-h_1\sum_{i=1}^na_{ij}(\lambda_i^k-\lambda_j^k)+\sigma^k\Delta\hat{P}_{a,i}^k,$$
$$\Delta\hat{P}_{a,i}^{k+1}=\Delta\hat{P}_{a,i}^{k}-z_{1}\sum_{i=1}^{n}a_{ij}(\Delta\hat{P}_{a,i}^{k}-\Delta\hat{P}_{a,j}^{k})+\Delta P_{i}^{k+1}-\Delta P_{i}^{k}.$$
From Fig. \ref{figc1c}, it can be seen that the settling time of MC and the estimated average power mismatch are about $45s$. Comparing with this, the schemes \eqref{3} and \eqref{19}, especially those of \eqref{19}, are superior to those of \cite{chenDistributedEconomicDispatch2021}, in terms of the convergence rate.
\subsection{Case 2. The Test on The BESS-oriented Plug-and-Play Function}
\par In a BESS network, certain BESSs may exit the energy supply sequence due to planning or unexpected events, such as low-level SoC, severe failures, etc. Therefore, this case is arranged to investigate the plug-and-play performance of
the two designed schemes and the scheme designed in \cite{chenDistributedEconomicDispatch2021}. Specifically, BESS 3 is isolated at $t = 30s$ and restored at $t = 70s$. The simulation results are shown in Figs. \ref{figc3a}-\ref{figc3c}.
\begin{figure}
	\centering
	\includegraphics[width=8cm]{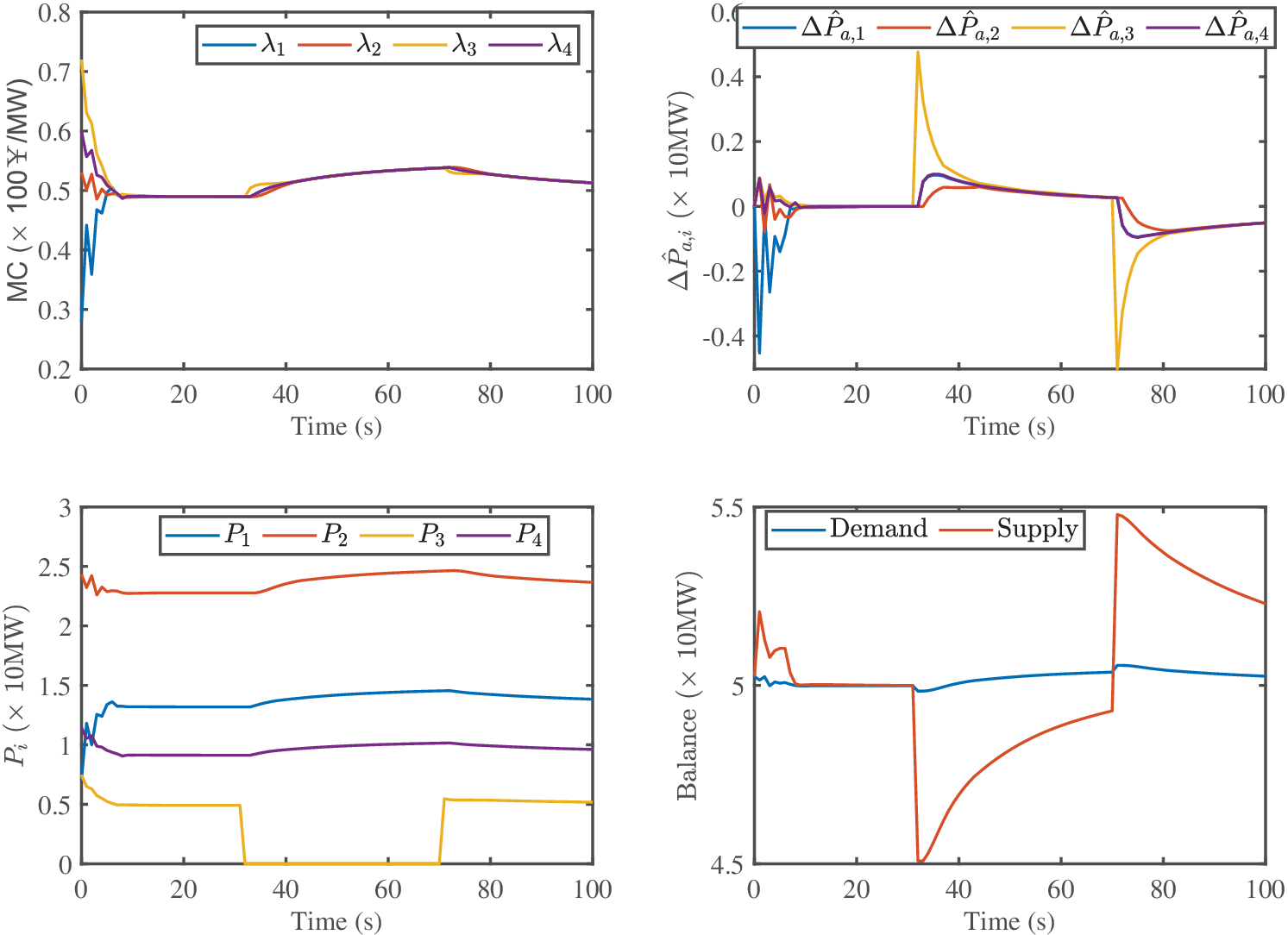} \caption{Simulation results of the designed scheme \eqref{3} under an isolated and connected BESS 3 in Case 2.\label{figc3a}}
\end{figure}
\begin{figure}
	\centering
	\includegraphics[width=8cm]{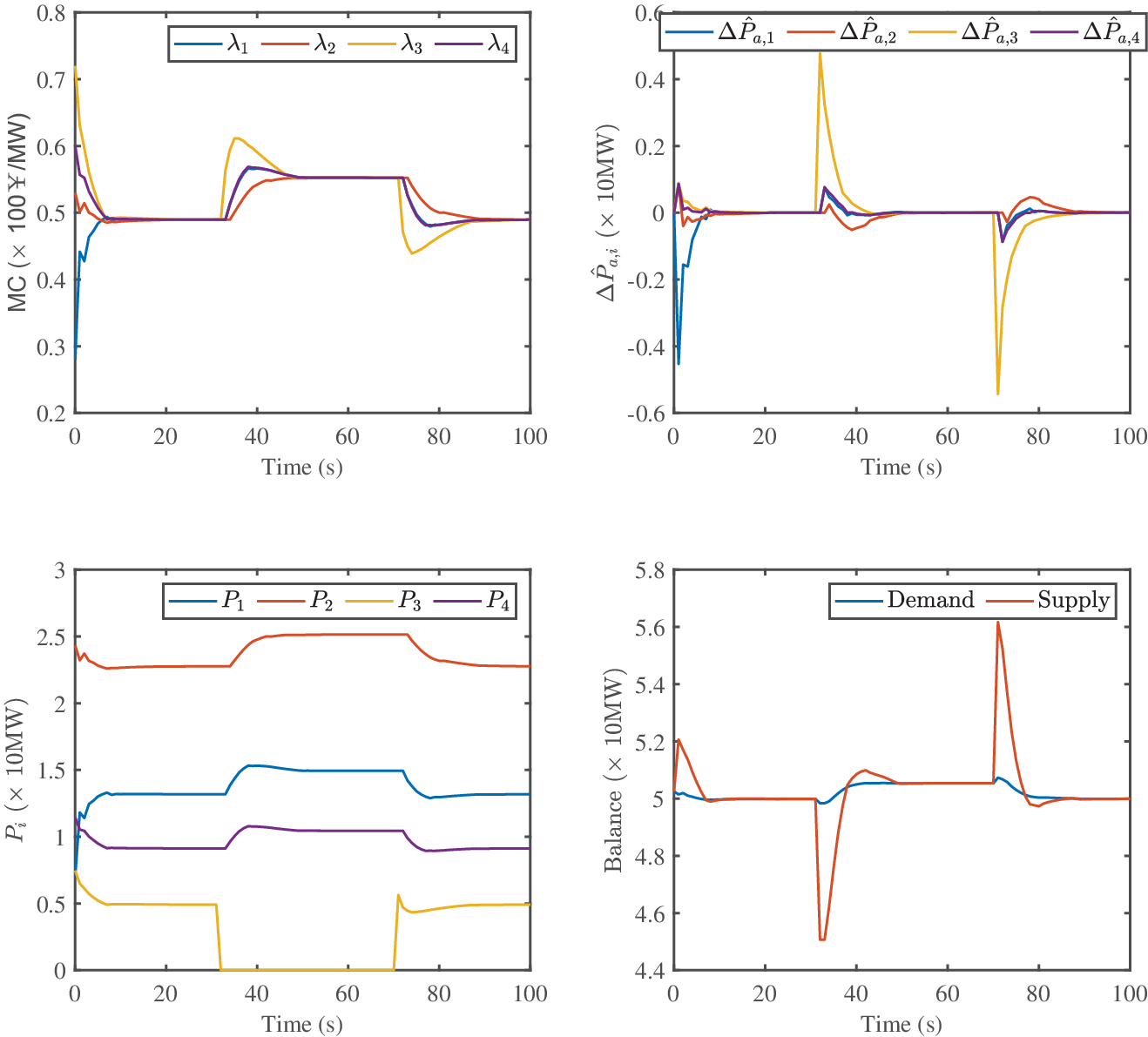} \caption{Simulation results of the designed scheme \eqref{19} under an isolated and connected BESS 3 in Case 2.\label{figc3b}}
\end{figure}
\begin{figure}
	\centering
	\includegraphics[width=8cm]{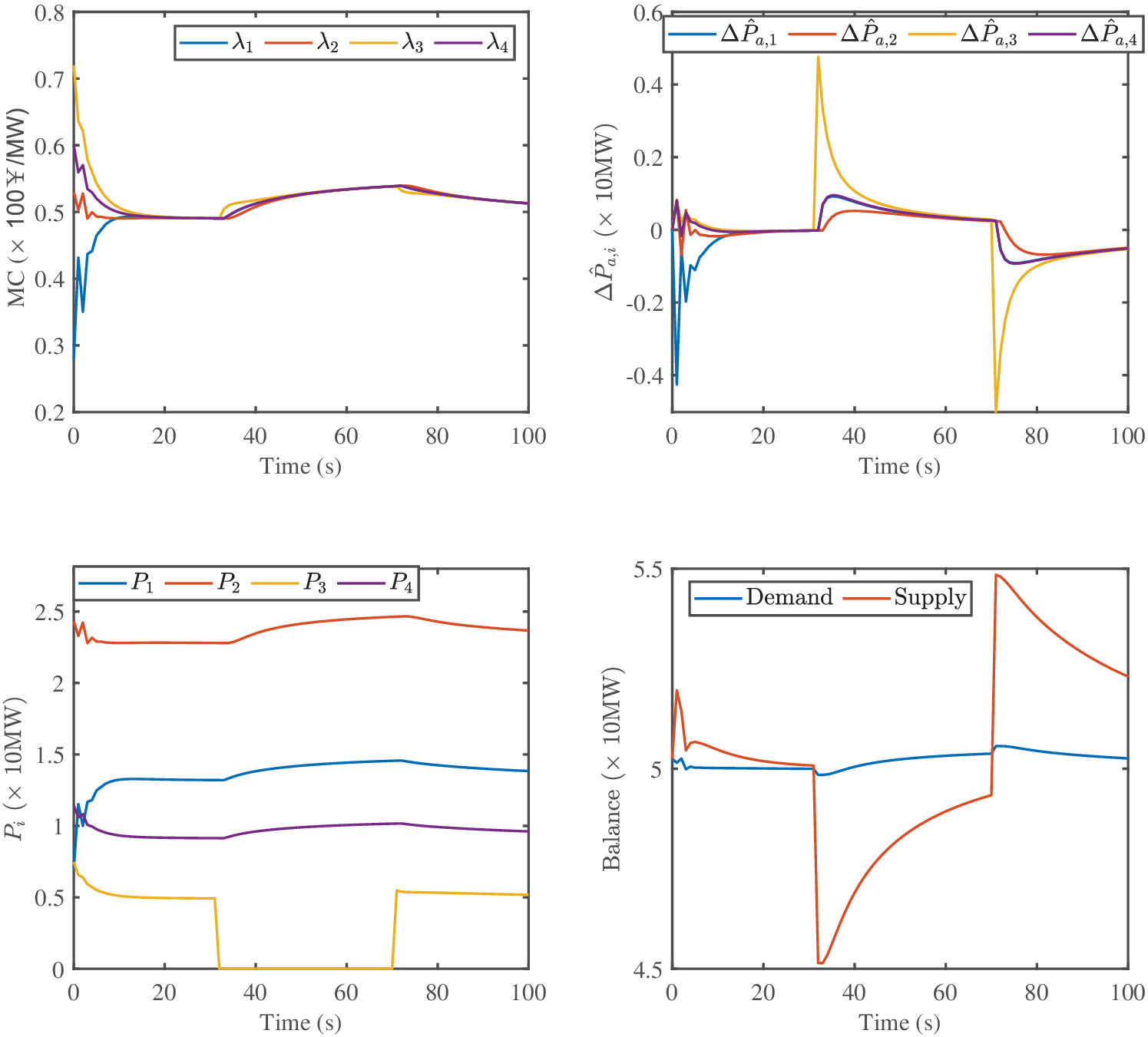} \caption{Simulation results of the designed scheme in \cite{chenDistributedEconomicDispatch2021} under an isolated and connected BESS 3 in Case 2.\label{figc3c}}
\end{figure}
\par From the simulation results, it can be seen that the performance of the two schemes in dealing with the plug-and-play of BESS 3 is different. Comparing the simulation results in Figs. \ref{figc3a} and \ref{figc3b}, it is not difficult to see that scheme \eqref{3} is on par with scheme \eqref{19} in terms of consensus rate, while scheme \eqref{3} is significantly inferior to scheme \eqref{19} in terms of convergence rate. This is because for the MC consensus scheme, the estimated average power mismatch is not introduced into the closed-loop error in scheme \eqref{3}. This further deteriorates the stability of the output power of each BESS and the supply and demand balance when BESS 3 is isolated or restored. In addition, the simulation results also indicate that as long as the total demand does not exceed the upper and lower limits of the network, the supply-demand balance can also be achieved.
\par Comparing Figs. \ref{figc3a} and \ref{figc3c}, under the scheme \eqref{3}, the consensus rate of MCs and the consensus rate of the estimated average power mismatch can be accelerated. However, the effect of the scheme \eqref{3} on improving convergence performance is negligible. In addition, the scheme \eqref{3} and the one in \cite{chenDistributedEconomicDispatch2021} perform very poorly in dealing with the exit and access of BESS 3, with a slow convergence rate leading to significant steady-state errors. Compared with Figs. \ref{figc3b} and \ref{figc3c}, under the scheme \eqref{19}, the convergence rate and consensus rate can be significantly accelerated, which makes the scheme have good BESS plug and play performance.
\par In summary, compared to the existing scheme, the consensus rate of MCs and the consensus rate of average power mismatch estimates have been appropriately accelerated under scheme \eqref{3}, but the convergence rate has not changed significantly. Under scheme \eqref{19}, the consensus and convergence performance is further greatly improved.
\subsection{Case 3. The Test on The Agent-oriented Plug-and-Play Function}
\begin{figure}
	\centering
	\includegraphics[width=8cm]{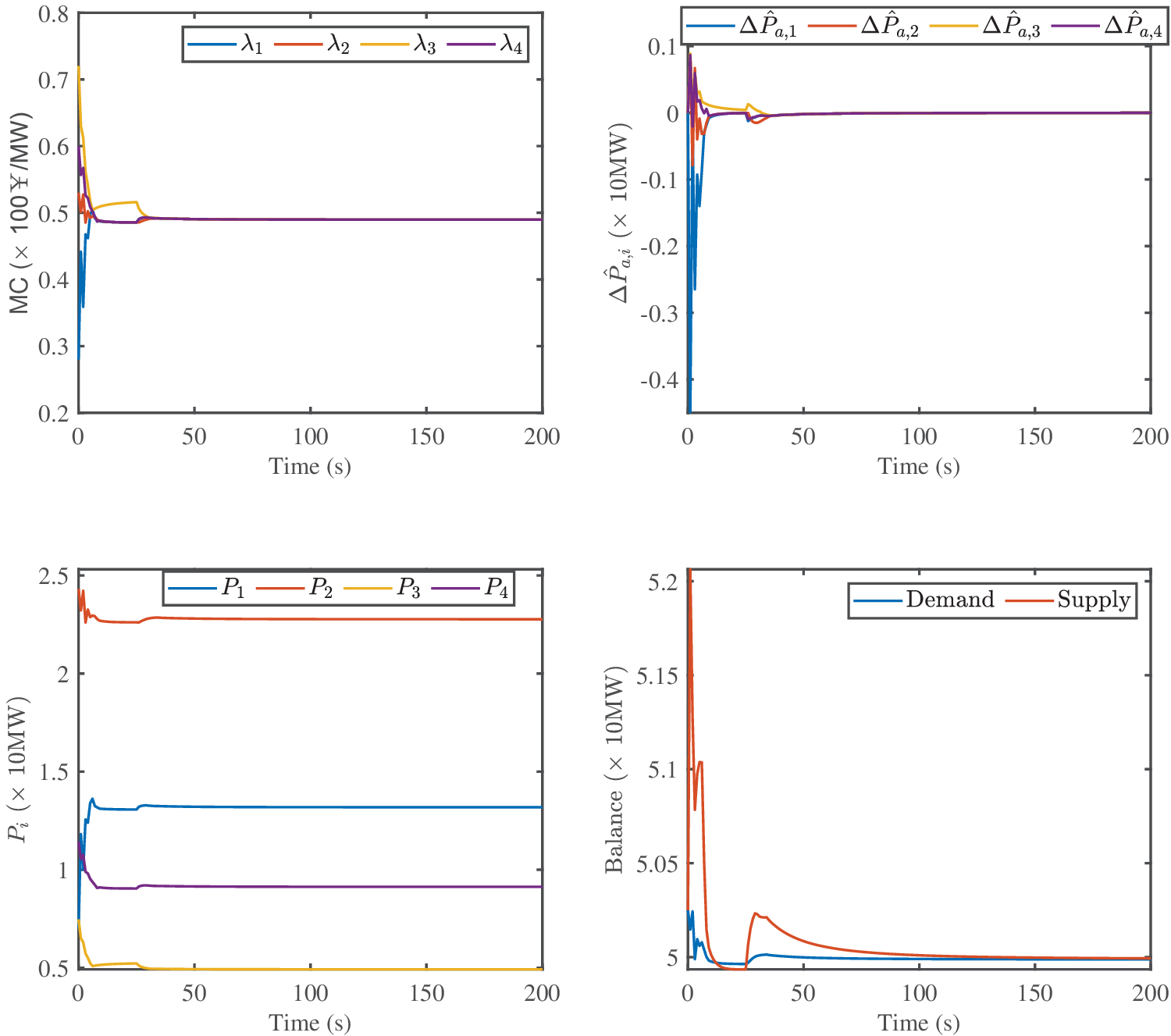} \caption{Simulation results of the designed scheme \eqref{3} under an isolated and connected agent 3 in Case 3.\label{figc4a}}
\end{figure}
\begin{figure}
	\centering
	\includegraphics[width=8cm]{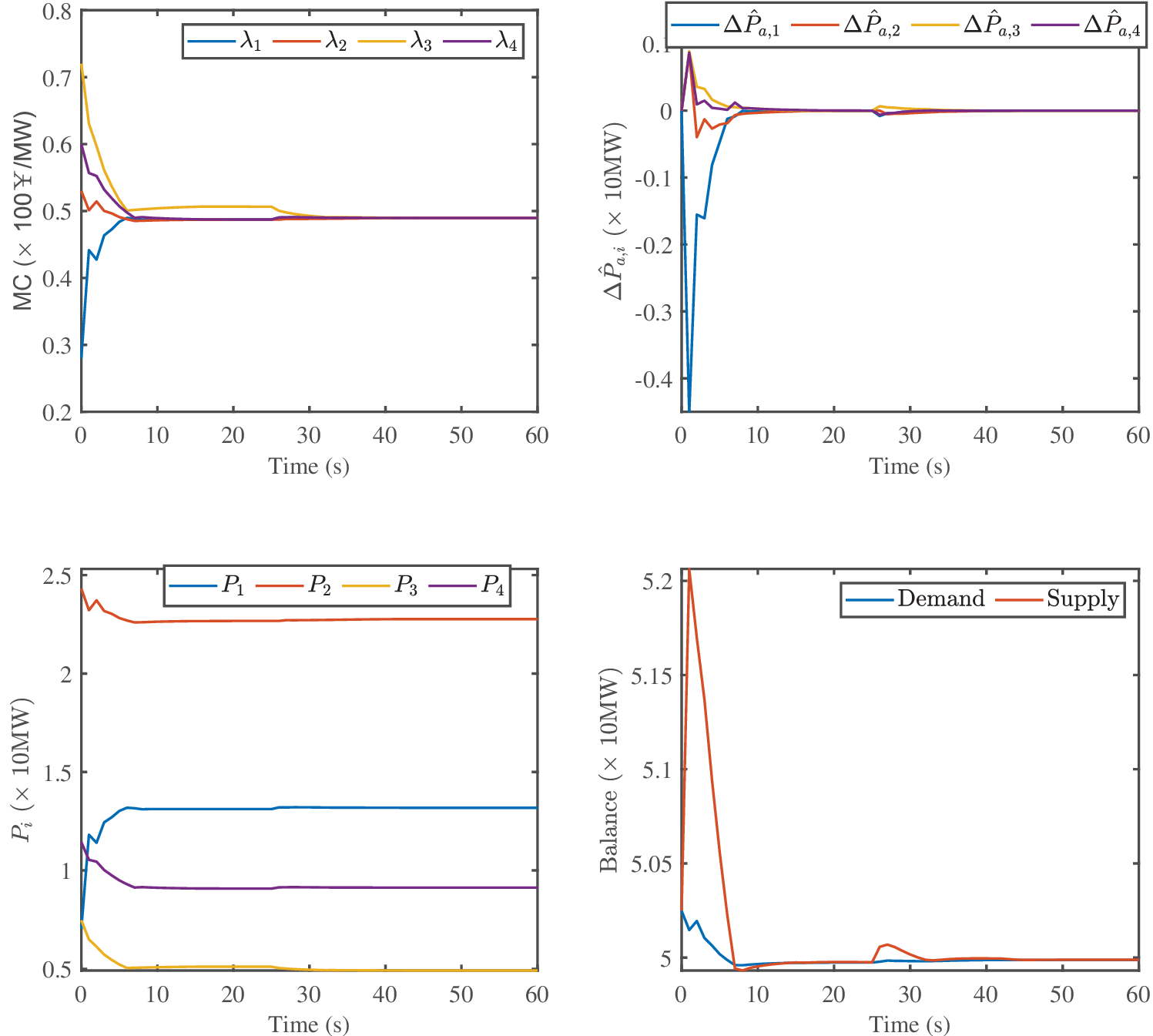} \caption{Simulation results of the designed scheme \eqref{19} under an isolated and connected agent 3 in Case 3.\label{figc4b}}
\end{figure}
\begin{figure}
	\centering
	\includegraphics[width=8cm]{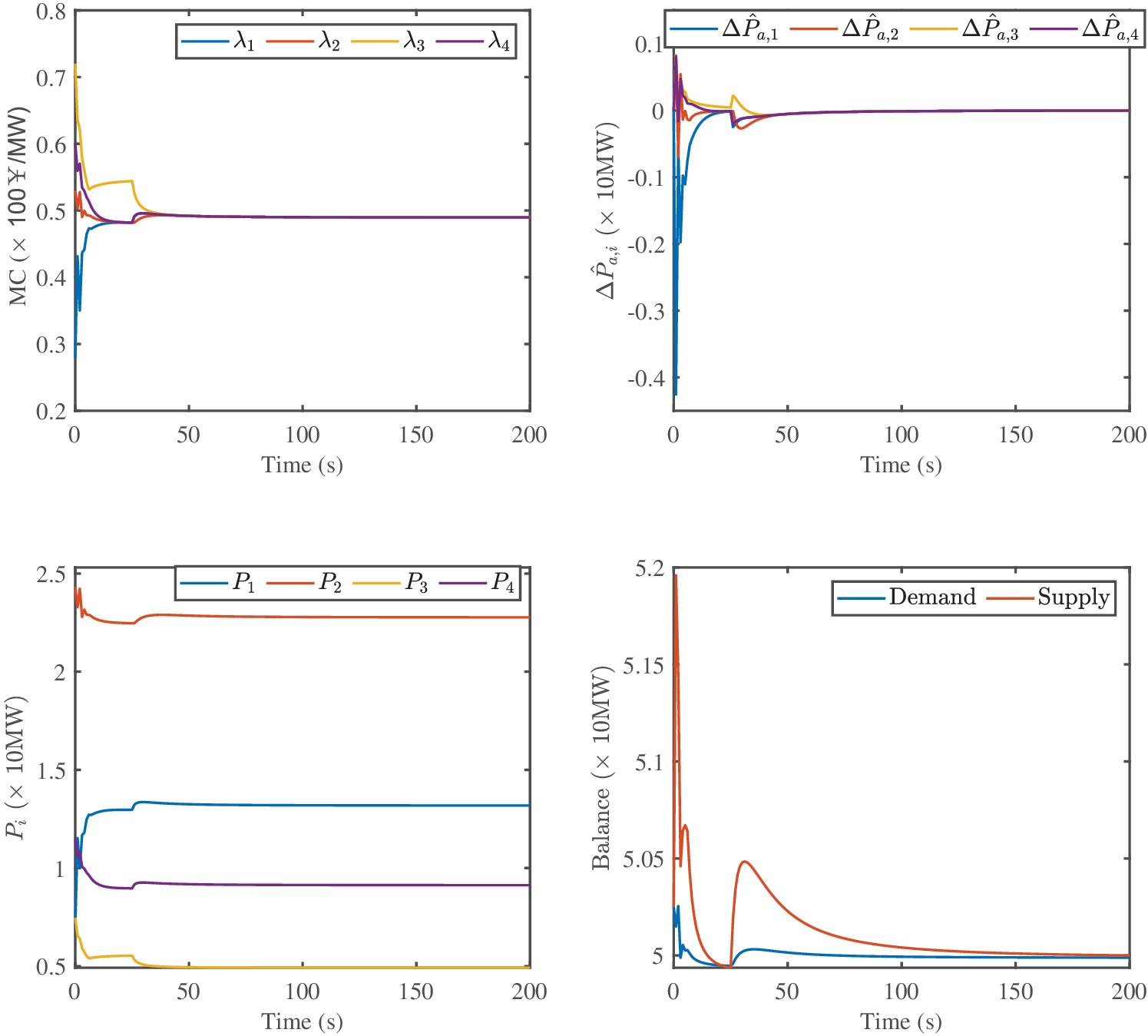} \caption{Simulation results of the designed scheme in \cite{chenDistributedEconomicDispatch2021} under an isolated and connected agent 3 in Case 3.\label{figc4c}}
\end{figure}
\par This case is arranged to investigate the effectiveness and progressiveness of the designed scheme in dealing with agent plug and play. Similar to Case 3, three ED schemes are simulated here. As for the simulation event, agent 3 refuses to communicate with others from $t=5s$ until it resumes communication at $t=25s$. The simulation results are shown in Figs. \ref{figc4a}-\ref{figc4c}.
\par From the simulation results, it can be seen that under schemes \eqref{3} and \eqref{19}, the access and exit of agent 3 will not affect the stability of each variable, and can ensure the balance of power supply and demand. However, some other significant differences have emerged. Compared with the existing scheme in \cite{chenDistributedEconomicDispatch2021}, under the drive of schemes \eqref{3} and \eqref{19}, the estimated average power mismatch and MC reach consensus faster. Furthermore, driven by scheme \eqref{19}, the convergence rate is accelerated, resulting in the scheme \eqref{19} being superior to \eqref{3}. Anyway, the scheme \eqref{3} outperforms the existing scheme in terms of consensus rate, while the scheme \eqref{19} outperforms the existing scheme and the scheme \eqref{3} in terms of consensus rate and convergence rate. In addition, the simulation results also indicate that as long as the communication structure is connected, MC consensus can be achieved, and the estimated average power mismatch can converge to 0.
\begin{figure}
	\centering
	\includegraphics[width=8cm]{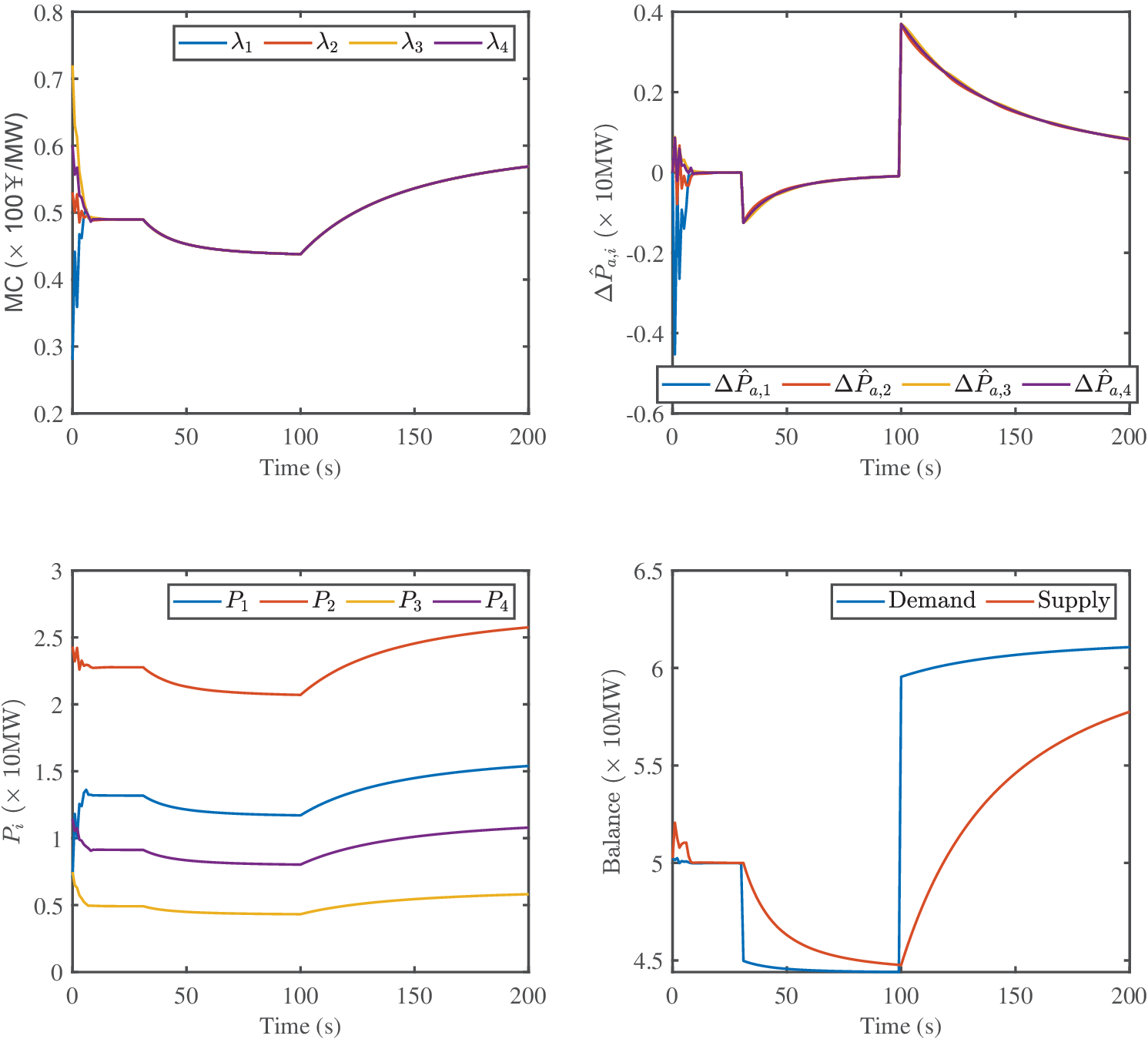} \caption{Simulation results of the designed scheme \eqref{3} under a load switching in Case 4.\label{figc5a}}
\end{figure}
\begin{figure}
	\centering
	\includegraphics[width=8cm]{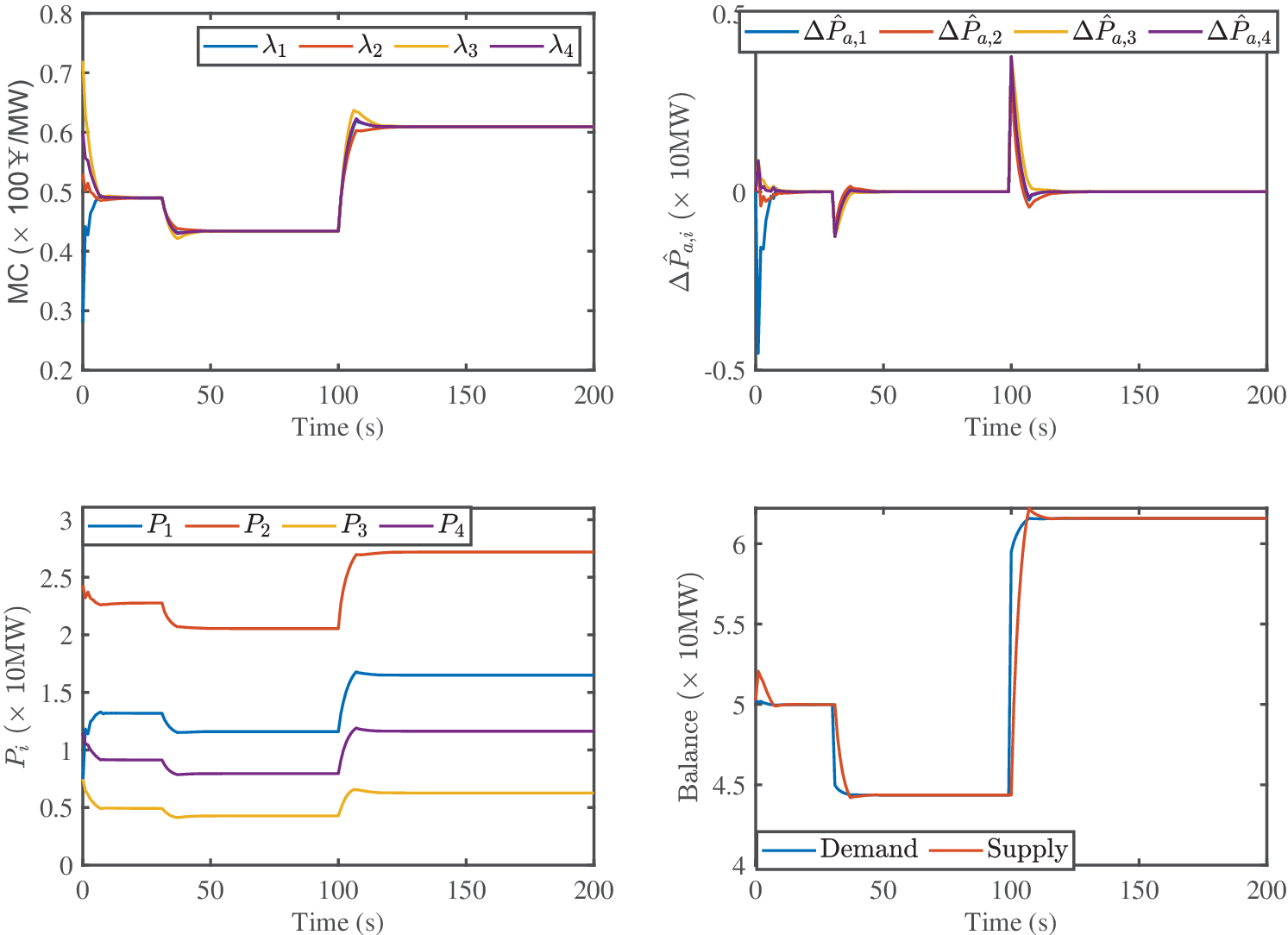} \caption{Simulation results of the designed scheme \eqref{19} under a load switching in Case 4.\label{figc5b}}
\end{figure}
\begin{figure}
	\centering
	\includegraphics[width=8cm]{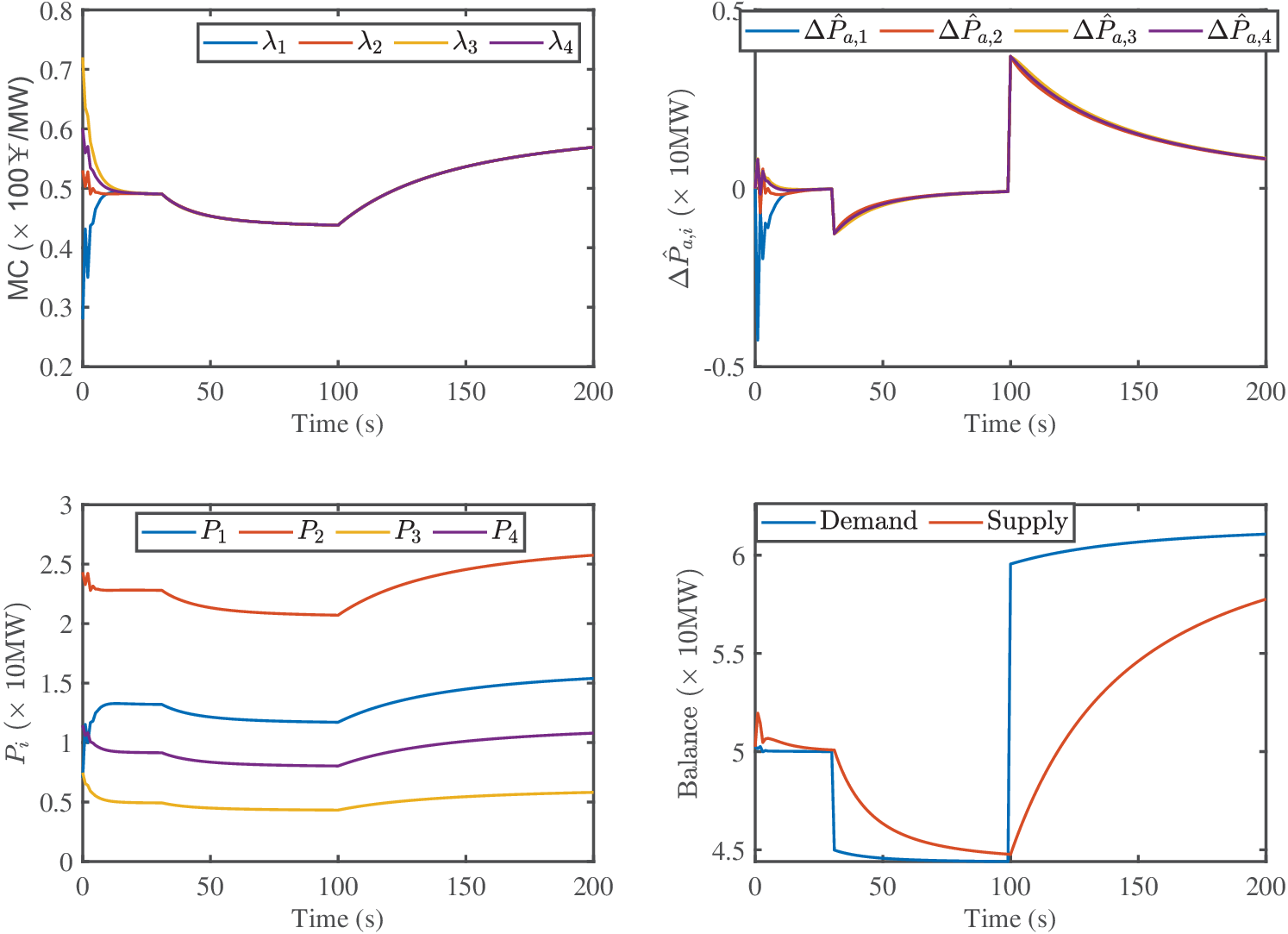} \caption{Simulation results of the designed scheme in \cite{chenDistributedEconomicDispatch2021} under a load switching in Case 4.\label{figc5c}}
\end{figure}
\subsection{Case 4. The Test on Load Switching}
\par To test the performance of the designed schemes and the scheme in \cite{chenDistributedEconomicDispatch2021} in handling load switching, this case study is arranged, and a load decrease and an increase are evenly distributed to each load node at $t=30s$ and $t=70s$, respectively. The simulation results are shown in Figs. \ref{figc5a}-\ref{figc5c}.
\par From Fig. \ref{figc5a}, it can be seen that MC consensus and the estimated average power mismatch consensus are less affected by load switching. However, the stability of both is greatly affected by the load switching. This is because in scheme \eqref{3}, the coupling item is not introduced into the closed-loop error of the MC consensus scheme. From Fig. \ref{figc5b}, it can be seen that the scheme \eqref{19} can handle load switching well, thereby ensuring MC consensus and the estimated average power mismatch consensus. At the same time, it can also ensure a balance between power supply and demand. It is worth mentioning that the scheme \eqref{19} has better stability than the scheme \eqref{3}, making it more suitable for handling load switching.
\par Compared with Fig. \ref{figc5c}, when dealing with load switching, scheme \eqref{3} has some advantages in consensus rate and control accuracy, while the convergence performance does not improve significantly. For scheme \eqref{19}, compared to the first two schemes, it shows strong advantages in terms of consensus rate, convergence rate, and control accuracy.
\begin{figure}
	\centering
	\includegraphics[width=7cm]{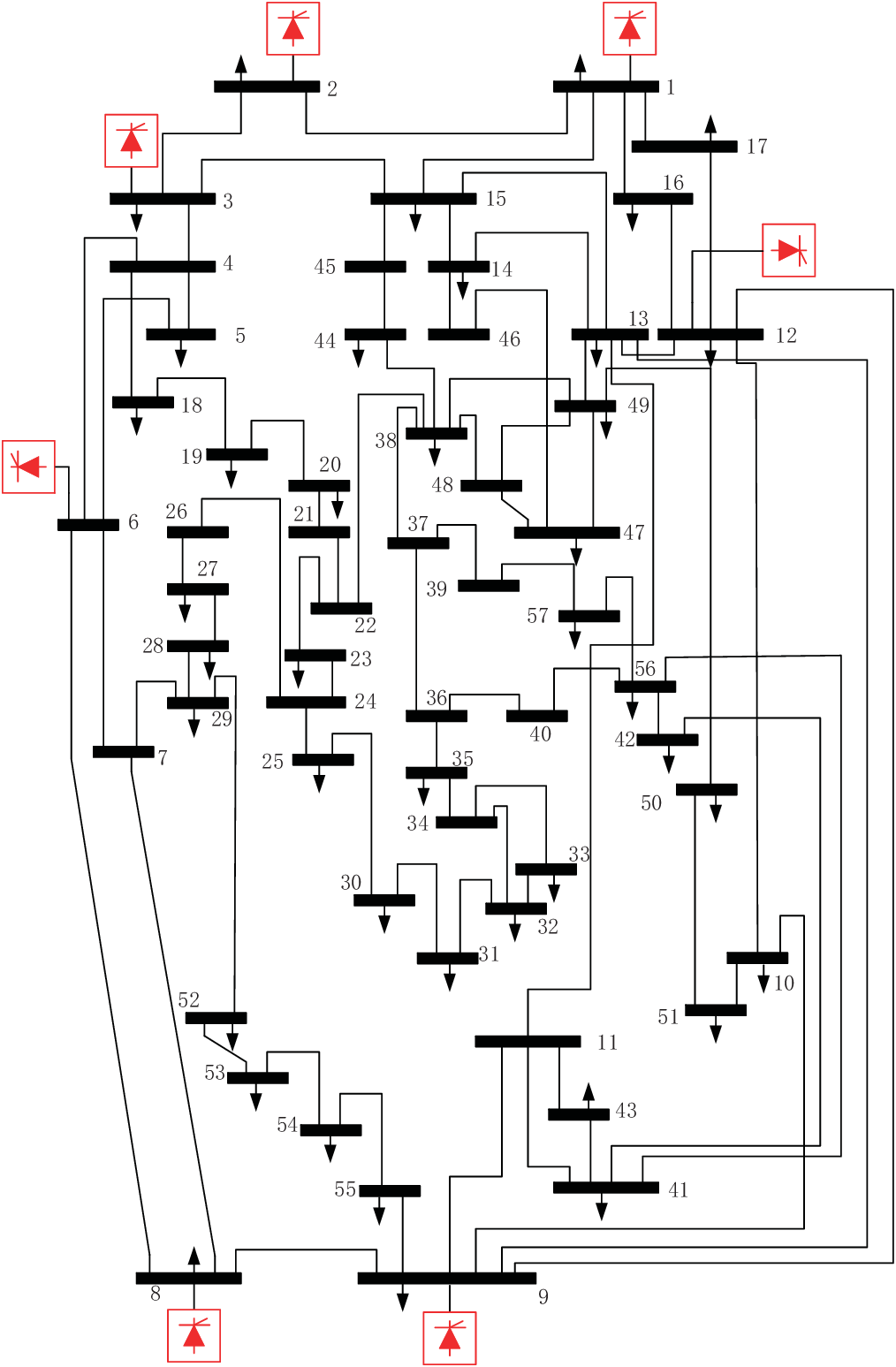} \caption{A IEEE-57 bus system.\label{fig6t}}
\end{figure}
\subsection{Case 5. The Test on A Large-scale Power System}
The interconnection of local power grids helps to enhance power supply reliability and improve the redundancy of microgrids. However, more complex power systems also pose new challenges to control algorithms. So, in this case, an IEEE-57 bus system \cite{PSEA2023} powered by BESSs, as shown in Fig. \ref{fig6t}, is used to test the designed solution. The simulation results are shown in Figs. \ref{figc6a}-\ref{figc6c}.
\begin{figure}
	\centering
	\includegraphics[width=8cm]{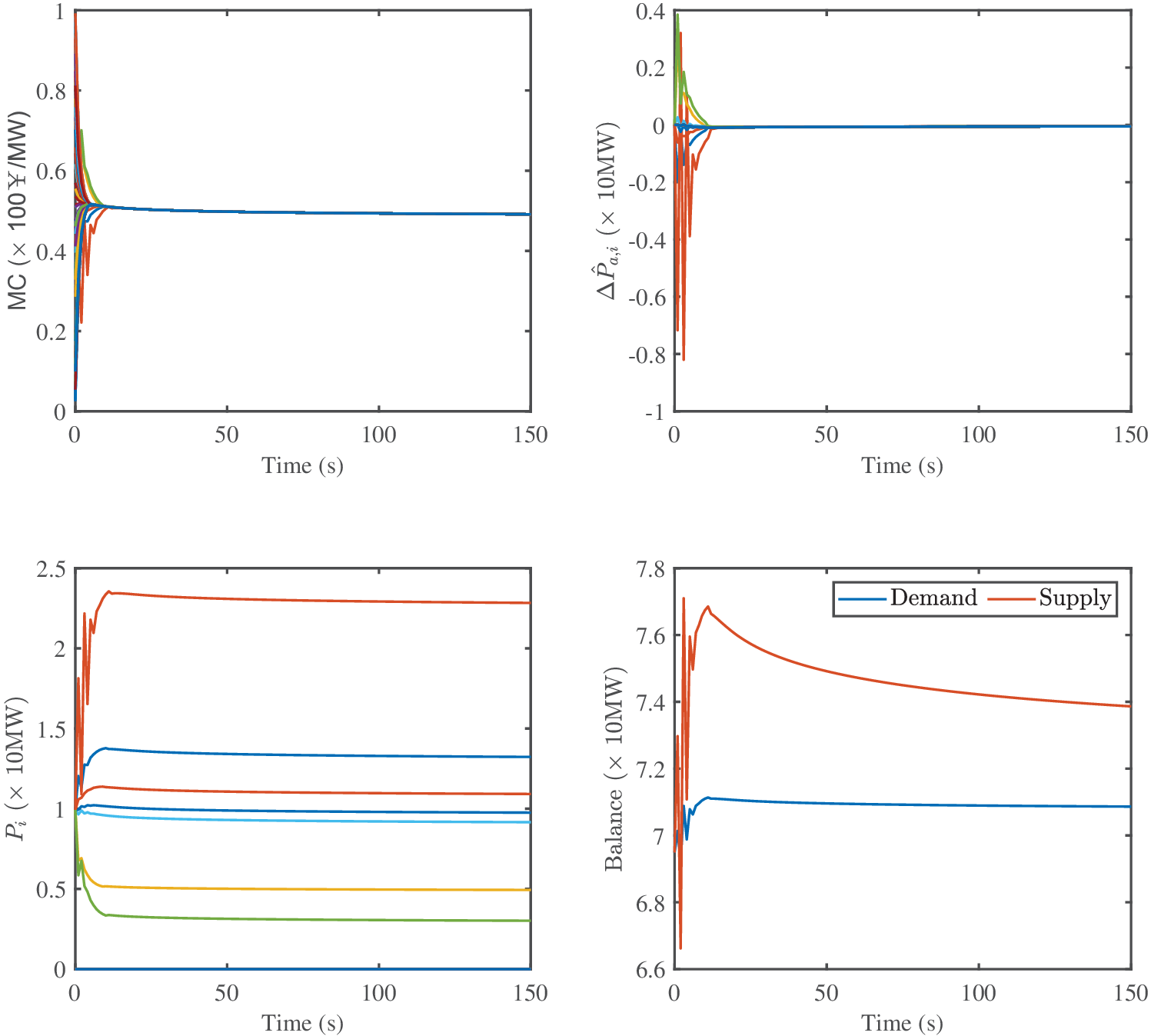} \caption{Simulation results of the designed scheme \eqref{3} on a modified IEEE-57 bus system in Case 5.\label{figc6a}}
\end{figure}
\begin{figure}
	\centering
	\includegraphics[width=8cm]{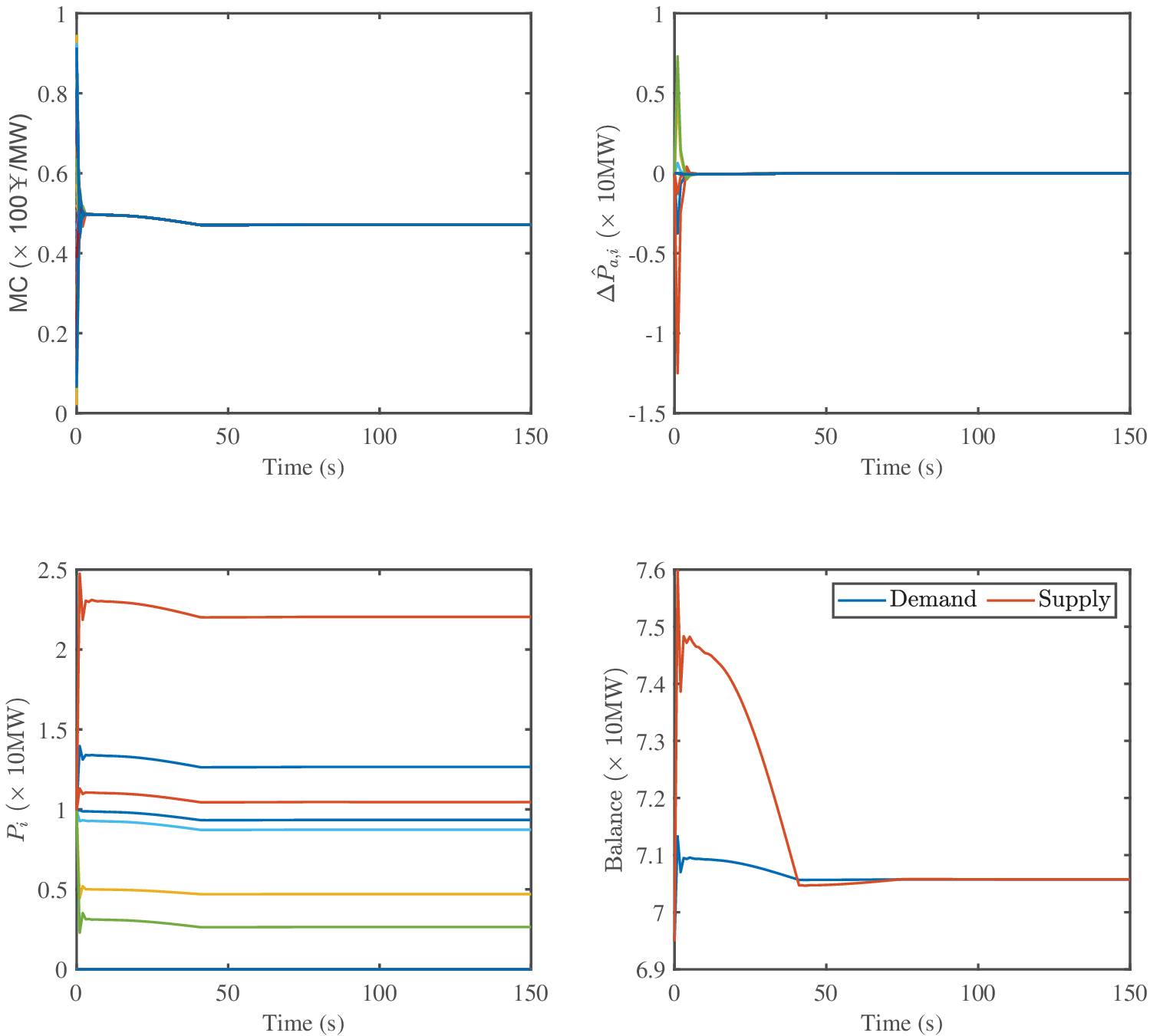} \caption{Simulation results of the designed scheme \eqref{19} on a modified IEEE-57 bus system in Case 5.\label{figc6b}}
\end{figure}
\begin{figure}
	\centering
	\includegraphics[width=8cm]{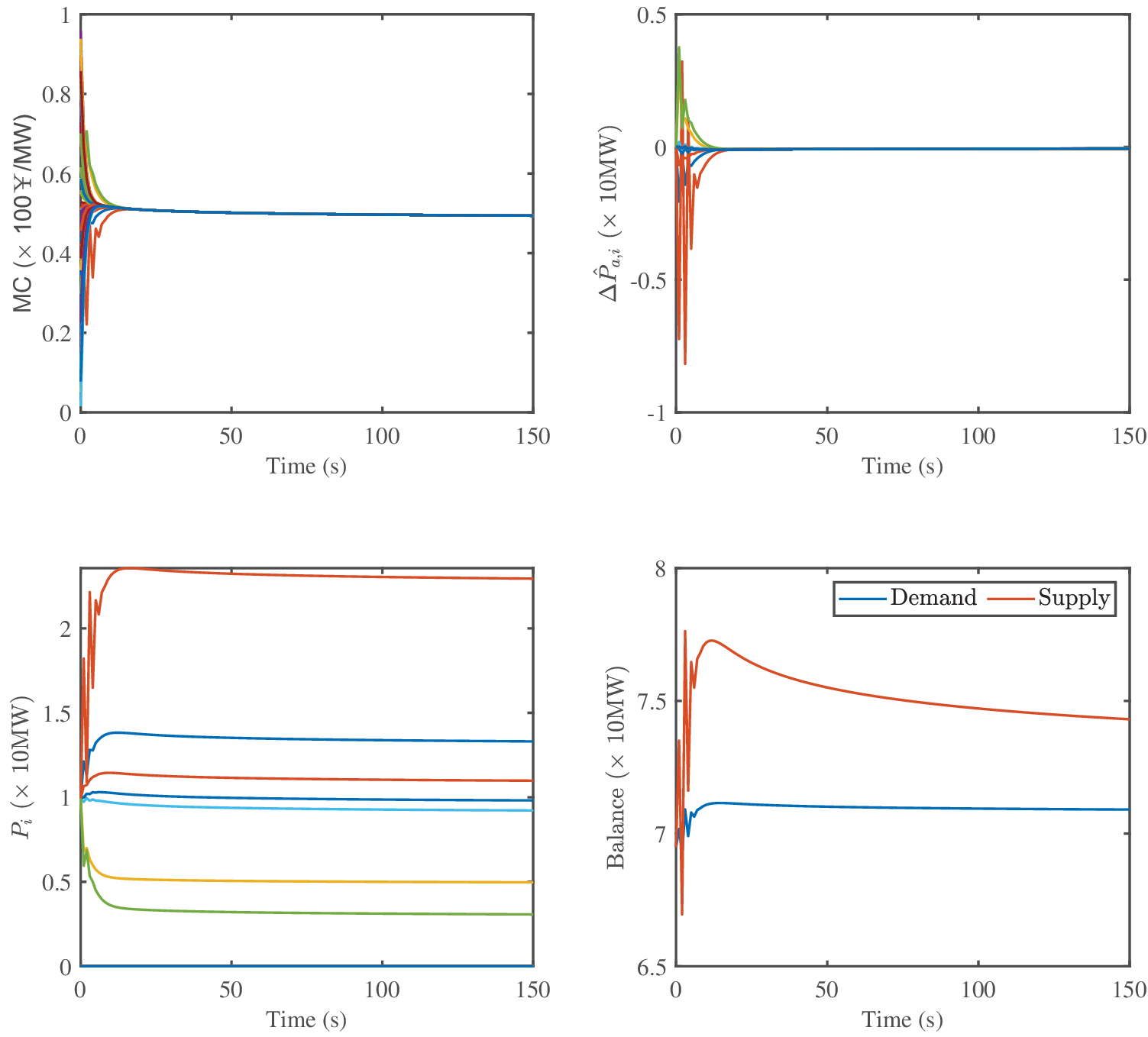} \caption{Simulation results of the designed scheme in \cite{chenDistributedEconomicDispatch2021} on a modified IEEE-57 bus system in Case 5.\label{figc6c}}
\end{figure}
\par For the scheme \eqref{3}, compared with the simulation result in Fig. \ref{figc6c}, it can be seen from Fig. \ref{figc6a} that MCs can reach consensus more quickly, and the estimated average power mismatch can converge to 0 faster. However, it is worth noting that comparing Fig. \ref{figc1a}, it can be seen that a large-scale system causes oscillations on MC and the estimated average power mismatch, resulting in oscillations in the output power of each BESS. In addition, from the simulation results, the control accuracy also deteriorates. In contrast, the scheme \eqref{19} performs very well. Compared to Fig. \ref{figc6b}, the scheme \eqref{19} does not cause significant oscillations in MC and the estimated average power mismatch, and causes the higher control accuracy.
\section{Conclusions}
\par In this article, two distributed ED schemes with PI+R protocols are designed for an isolated BESS network. From the simulation results, under the condition that the output power of each BESS does not exceed the limit, the operating expenses of the BESS network are the lowest, and the supply-demand balance is kept well. In addition, compared to an existing scheme in \cite{chenDistributedEconomicDispatch2021}, the consensus rates of MC and the estimated average power mismatch have been improved by at least 30\%. In the face of load switching, isolation of a BESS/agent and scalability, the design schemes still exhibit good performance. Compared with the existing scheme, the designed schemes, especially the second one, have strong advantages in consensus rate, control accuracy, and dynamic performance, and have stronger scalability and robustness in dealing with BESS/agent plug and play, load switching, and wide area systems. The PI+R controller designed has a simple structure, strong engineering feasibility, and good application prospects. Subsequent work will start with communication delay, privacy protection, and more complex line loss models to further improve the existing solutions.
\subsection*{Conflict of Interest}
There are no conflict of interest.
\subsection*{Authors' Contributions}
Yalin Zhang: Conceptualization, Methodology, Software, Writing - Original Draft, Visualization.
\par Zhongxin Liu: Validation, Resources, Writing - Review $\&$ Editing.
\par Zengqiang Chen: Supervision.
\subsection*{Funding }
This work is supported by the National Natural Science Foundation of China (Grant No. 52477133).
\bibliographystyle{abbrv}
\bibliography{re1}
\biography{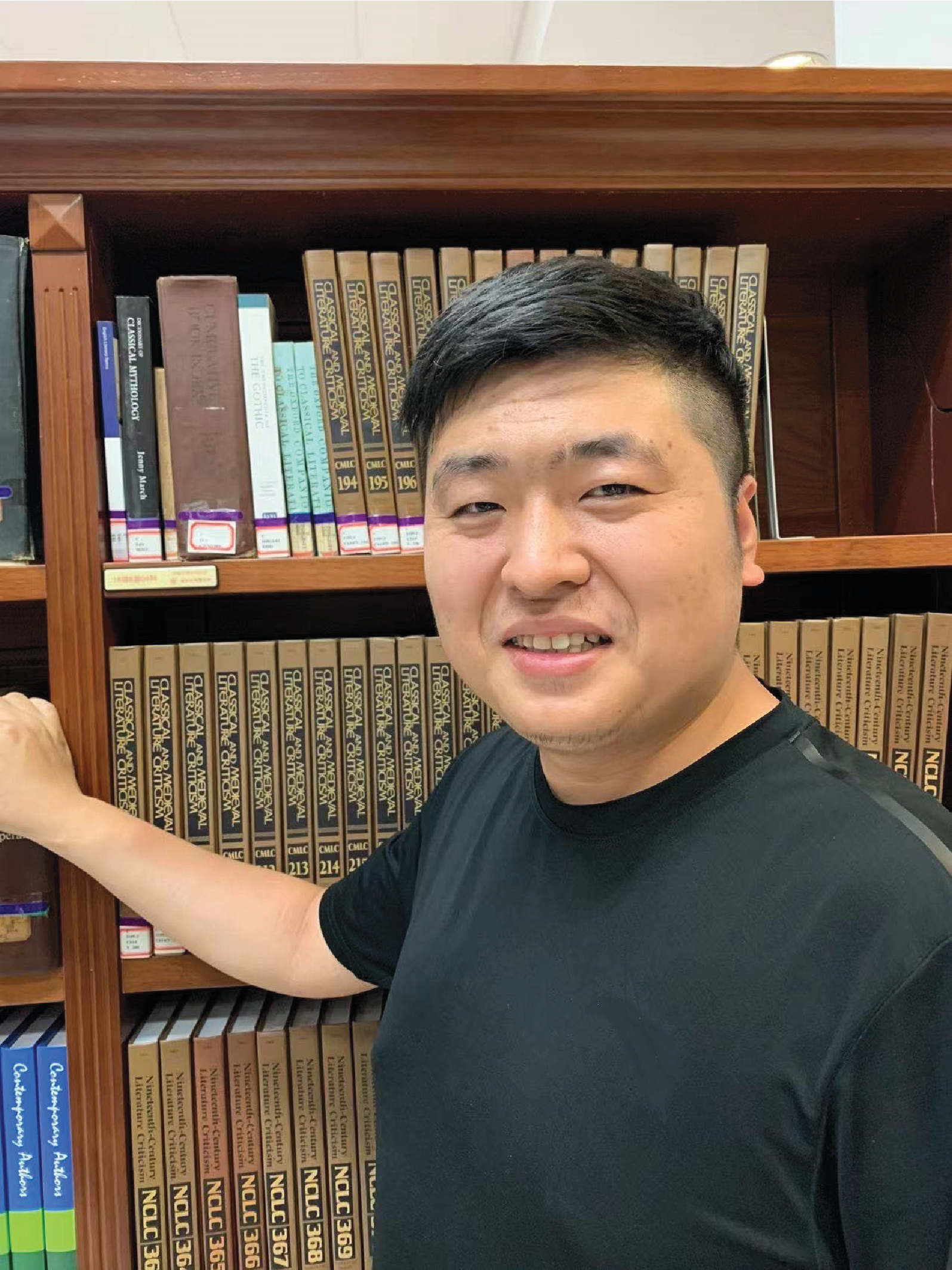}{Yalin Zhang}{received his B.S. and M.S. degrees in Automation and Control Science and Engineering from Henan Polytechnic University in 2016 and 2020, respectively. He also received his PhD in Control Science and Engineering from Nankai University in 2025.
\par He is currently working as a postdoctoral assistant researcher at the School of Electrical Engineering, Zhejiang University. He has published over ten papers in IEEE Transactions and served as a reviewer for several high-level journals. In 2020, he obtained an excellent master's thesis at the university level. His research interests include distributed collaborative control, learning and optimization based on multi-agent systems, complex network theory, and their application in distributed collaborative control and economic dispatch of smart grids.}
\biography{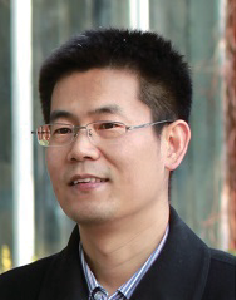}{Zhongxin Liu}{received his B.S. degree in automation and the Ph.D. degree in control theory and control engineering from Nankai University, Tianjin, China, in1997 and 2002, respectively. 
\par He has been with Nankai University, where he is currently a Professor with the Department of Automation. His research interests include multi-agent systems, nonlinear control theory, networked control system, as well as complex network theory and its application.}
\biography{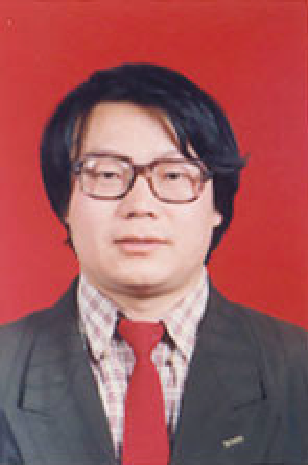}{Zengqiang Chen}{received his B.S. degree in mathematics and the M.S. and Ph.D. degrees in control theory and control engineering from Nankai University, Tianjin, China, in 1987, 1990, and 1997, respectively. 
\par He has been with Nankai University, where he is currently a Professor with the Department of Automation. His research interests include predictive control technology, complex network system, chaotic system theory and its application in information security.} 
\clearafterbiography
\relax 
\end{document}